\documentclass[11pt,letterpaper]{article}
\usepackage[margin=1in]{geometry}
\usepackage{amsmath,amssymb,amsfonts,setspace,url,mathrsfs}
\usepackage{bm,enumerate}
\usepackage[boxed,linesnumbered,noend]{algorithm2e}
\DontPrintSemicolon
\let\oldnl\nl
\newcommand{\nonl}{\renewcommand{\nl}{\let\nl\oldnl}}
\usepackage{algpseudocode}
\usepackage{latexsym}
\usepackage{amsthm}
\usepackage{subfig}
\usepackage{float}
\usepackage{stmaryrd}
\usepackage{fancybox}
\usepackage{bigstrut,array,multirow}
\usepackage{here}
\usepackage{color}
\usepackage{url}
\usepackage{hyperref}
\usepackage[capitalise]{cleveref}
\usepackage{tikz}

\usepackage{mathpazo}

\newcommand{\abs}[1]{\vert #1 \vert}
\newcommand{\gen}[1]{\langle #1 \rangle}
\newcommand{\ket}[1]{\vert #1 \rangle}

\newcommand{\bul}{\ast}
\newcommand{\Sym}{\mathrm{Sym}}

\newcommand{\myK}{\Sigma}
\newcommand{\ns}{\unlhd}

\newcommand{\mybar}[1]{\lambda}

\newcommand{\group}[1]{\mathrm{gr}(#1)}
\newcommand{\bbgroup}{\mathcal{G}}
\newcommand{\bbgroupb}{\mathcal{H}}

\newcommand{\Pp}{\mathcal{R}}
\newcommand{\poly}{\mathrm{poly}}

\newcommand{\QMA}{\mathrm{QMA}}
\newcommand{\coQMA}{\mathrm{coQMA}}
\newcommand{\QCMA}{\mathrm{QCMA}}
\newcommand{\coQCMA}{\mathrm{coQCMA}}

\newcommand{\AM}{\mathrm{AM}}
\newcommand{\coAM}{\mathrm{coAM}}
\newcommand{\MA}{\mathrm{MA}}

\newcommand{\NP}{\mathrm{NP}}
\newcommand{\coNP}{\mathrm{coNP}}

\newcommand{\seed}{\lambda}
\newcommand{\valid}{\mathsf{Valid}(\seed)}
\newcommand{\good}{\mathsf{Good}}
\newcommand{\majproc}{\mathsf{Maj}}
\newcommand{\htest}{\mathsf{Iso Check}}

\newcommand{\polylog}{\mathrm{polylog}}

\newcommand{\soc}[1]{\mathrm{Soc^*}(#1)}
\newcommand{\sol}[1]{\mathrm{Sol}(#1)}
 
\newcommand{\pker}[1]{\mathrm{Pker}(#1)}
\newcommand{\tgroup}{\{e\}}

\newcommand{\Qq}{\mathsf{ReeMemb}}

\newcommand{\GL}{\mathrm{GL}}
\newcommand{\Ree}{\mathrm{R}(q)}
\newcommand{\oRee}[1]{^2\mathrm{G}_2(#1)}

\newcommand{\field}{\mathbb{F}}

\newcommand{\triv}{\{e\}}

\newtheorem{theorem}{Theorem}[section]
\newtheorem{corollary}{Corollary}[section]
\newtheorem{proposition}{Proposition}[section]
\newtheorem{definition}{Definition}[section]

\newtheorem{fact}{Fact}
\newtheorem{conjecture}{Conjecture}
\newtheorem{oproblem}{Open Problem} 

\newtheorem{claim}{Claim}
\newtheorem{lemma}{Lemma}[section]
\newenvironment{proof-sketch}{\trivlist\item[]\emph{Brief proof sketch}:}%
{\unskip\nobreak\hskip 1em plus 1fil\nobreak$\Box$
\parfillskip=0pt%
\endtrivlist}
\begin{document}

\title{Group Order is in $\QCMA$}
\author{
Fran{\c c}ois Le Gall\\
Graduate School of Mathematics\\
Nagoya University\\legall@math.nagoya-u.ac.jp
\and Harumichi Nishimura\\
Graduate School of Informatics\\
Nagoya University\\
hnishimura@i.nagoya-u.ac.jp
\and Dhara Thakkar\\
Graduate School of Mathematics\\
Nagoya University\\thakkar\_dhara@math.nagoya-u.ac.jp}
\date{}
\maketitle
\begin{abstract}
In this work, we show that verifying the order of a finite group given as a black-box is in the complexity class $\QCMA$. This solves an open problem asked by Watrous in 2000 in his seminal paper on quantum proofs and directly implies that the Group Non-Membership problem is also in the class $\QCMA$, which further proves a conjecture proposed by Aaronson and Kuperberg in 2006. Our techniques also give improved quantum upper bounds on the complexity of many other group-theoretical problems, such as group isomorphism in black-box groups.  
\end{abstract}
\section{Introduction}
\subsection{Background}
\paragraph{QMA and Group Non-Membership.}
The complexity class $\QMA$ (Quantum Merlin-Arthur) is one of the central complexity classes in quantum complexity theory. This class was first proposed by Knill \cite{knill1996} and Kitaev \cite{Kitaev99} as a natural quantum analogue of the classical class $\NP$ (or, more precisely, its randomized version called $\MA$), in which an all-powerful prover (named Merlin) sends a quantum proof to a verifier (named Arthur) who can perform bounded-error polynomial-time quantum computation. In 2000, Watrous \cite{WatrousFOCS00} established its power by showing that several group-theoretic problems are in $\QMA$ in the black-box setting. 

The concept of black-box group was first introduced (in the classical setting) by Babai and Szemerédi \cite{Babai+FOCS84} to describe group-theoretic algorithms in the most general way, without depending on how elements are concretely represented and how group operations are implemented. In a black-box group, each group element is represented by a binary string and each group operation (group multiplication and inversion) is implemented using an oracle. Any efficient algorithm in the black-box group model thus gives rise to an efficient concrete algorithm when oracle operations can be replaced by efficient procedures, which can be done for many natural group representations, including permutation groups and matrix groups. In the quantum setting introduced by Watrous~\cite{WatrousFOCS00} and further investigated in several further works \cite{Aaronson+ToC07,Ivanyos+03,LeGall+MFCS18,LeGallSTACS10,WatrousSTOC01}, the oracles should be able to handle quantum superpositions. Additionally, in the quantum setting, all these works assume that the group has unique encoding, i.e., each element should be encoded using a unique string (without unique encoding, even the most basic quantum primitives, such as computing the order of one element of the group, cannot be implemented).

The central problem considered in \cite{WatrousFOCS00} is the Group Non-Membership problem defined below (where, for any elements $g_1,\ldots,g_k$ of a group, we denote by $\gen{g_1,\ldots,g_k}$ the subgroup generated by $g_1,\ldots,g_k$): 

\begin{center}
\underline{Group Non-Membership}\\\vspace{-3mm}
\begin{flushleft}
\begin{tabular}{ll}
Instance: & Group elements $g_1,\ldots,g_k$ and $h$ in some finite group $\bbgroup$.\\
Question: & Is $h$ outside the group generated by $g_1,\ldots,g_k$ (i.e., is $h\notin G$ with $G=\gen{g_1,\ldots,g_k}$)\,?
\end{tabular}
\end{flushleft}
\end{center}

Group Non-Membership is significantly more challenging than its complement, Group Membership, which asks if $h\in\gen{g_1,\ldots,g_k}$: while Ref.~\cite{Babai+FOCS84} showed that Group Membership is in the class $\NP$, the best known classical upper bound for Group Non-Membership is $\AM$ (the class of problems that can be solved by a constant-round interactive proof system with public coins), by Babai \cite{Babai92}. We refer to Figure~\ref{fig} for an illustration of the relations between the complexity classes discussed in this paper.

\begin{figure}[!thbp]
\begin{center}
\begin{tikzpicture}[main/.style = {circle}, scale=1.0] 
\node[main](NP) at (0,0) {{\large $\NP$}};
\node[main](MA) at (0,1.5) {{\large$\MA$}};

\node[main](QCMA) at (4.2,2.8) {{\large$\QCMA\, (=\QCMA_1$)}};
\node[main](QCMAt) at (3,3.1) {};
\node[main](QCMAb) at (3,2.5) {};

\node[main](QMA1) at (3,4.3) {{\large$\QMA_1$}};
\node[main](QMA1t) at (3,4.6) {};
\node[main](QMA1b) at (3,4) {};

\node[main](QMA) at (3,6) {{\large$\QMA$}};
\node[main](QMAb) at (3,5.7) {};

\node[main](AM) at (0,7) {{\large$\AM$}};
\node[main](SIGMA) at (-3,4) {{\large$\Sigma_2^{\mathrm{P}}$}};

\draw[thick] (NP) -- (MA);
\draw[thick] (MA) -- (QCMAb);
\draw[thick] (QCMAt) -- (QMA1b);
\draw[thick] (QMA1t) -- (QMAb);
\draw[thick] (MA) -- (SIGMA);
\draw[thick] (MA) -- (AM);



\end{tikzpicture}
    \caption{Known relations between the main complexity classes discussed in this paper. The inclusion $\MA\subseteq\Sigma_2^{\mathrm{P}}$ was shown by Babai \cite{Babai85}. The equality $\QCMA=\QCMA_1$ was shown by Jordan, Kobayashi, Nagaj and Nishimura \cite{Jordan+12}. All the other relations follow directly from the definitions.}
    \label{fig}
\end{center}
\end{figure}
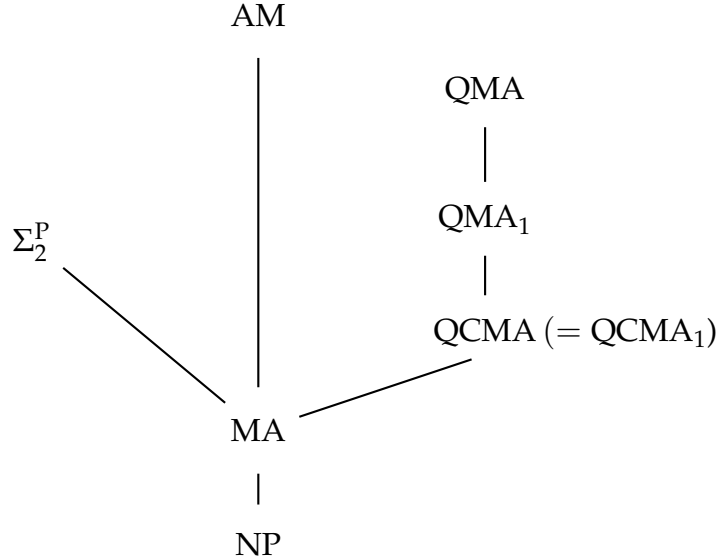

Watrous~\cite{WatrousFOCS00} showed that Group Non-Membership is in $\QMA$. 
To prove this result, the quantum proof received from Merlin is the quantum superposition of all the elements in the group $G=\gen{g_1,\ldots,g_k}$. Arthur checks that the quantum proof is valid (by checking that the quantum state is invariant under multiplication by $g_1,\ldots,g_k$) and then checks that this quantum state is mapped to an orthogonal state when each element in the superposition is multiplied by $h$ (which guarantees that $h\notin G$). The key feature of this protocol is that it uses a quantum proof. Indeed, Watrous \cite{WatrousFOCS00} also showed that there exist black-box groups for which Group Non-Membership is not in $\MA$, which shows that $\QMA$ is strictly more powerful than $\MA$ in the black-box setting (from a complexity-theoretic perspective, this can be interpreted as an oracle separation between $\QMA$ and $\MA$). Additionally, Watrous~\cite{WatrousFOCS00} showed that several additional group-theoretic problems (discussed later) are also contained in $\QMA$ via (fairly straightforward) reductions to Group Non-Membership. Besides Group Non-Membership being one of the most fundamental problems in $\QMA$, Watrous' protocol is often used in educational material to illustrate the power of quantum proofs (see, e.g., \cite{Aaronson16,ODonnel15,TCS-068,deWolf19}), due to its simplicity.
We also mention a later result by Grilo, Kerenidis and Sikora \cite{Grilo+16}, which showed that Group Non-Membership is actually in $\QMA_1$, the one-sided version of $\QMA$.

\paragraph{Group Non-Membership and QCMA.}
An important subclass of $\QMA$ is the class $\QCMA$ corresponding to problems where the proof is classical. One of the main open problems in quantum complexity theory, first posed by Aharonov and Naveh \cite{Aharonov+02}, is whether there exists a classical oracle separating $\QMA$ and $\QCMA$ (we refer to \cite{Aaronson+ToC07,Ben-David+24,Fefferman+18,li+ITCS24,Liu+24,Natarajan+24,Zhandry24} for partial progress).
In 2006, Aaronson and Kuperberg \cite{Aaronson+ToC07} showed that Group Non-Membership is actually in the class $\QCMA$ under some group-theoretic assumptions, which gives evidence that Group Non-Membership is not a good candidate for a separation between $\QMA$ and $\QCMA$. Aaronson and Kuperberg further conjectured that Group Non-Membership is actually in $\QCMA$ unconditionally:
\begin{conjecture}[\cite{Aaronson+ToC07}]\label{conj}
Group Non-Membership is in $\QCMA$.
\end{conjecture}
No progress has been made on this conjecture since 2006.

\paragraph{Group Order Verification.}
As explained above, Group Non-Membership is a fundamental task in group theory. 
An even more powerful primitive is computing the order of a group. For a group~$G$, we write its order (i.e., the number of elements in $G$) as $\abs{G}$.
We introduce the decision version of this problem as follows:
\begin{center}
\underline{Group Order Verification}\\\vspace{-3mm}
\begin{flushleft}
\begin{tabular}{ll}
Instance: & Group elements $g_1,\ldots,g_k$ in some finite group $\bbgroup$, and a positive integer $m$.\\
Question: & 
Is the order of the group generated by $g_1,\ldots,g_k$ equal to $m$ (i.e., is $\abs{G}=m$\\
&\hspace{101mm}
with $G=\gen{g_1,\ldots,g_k}$)\,?
\end{tabular}
\end{flushleft}
\end{center}\vspace{2mm}

Group Non-Membership reduces to Group Order Verification since $h\notin G$ if and only if $\abs{G}\neq \abs{\gen{g_1,\ldots,g_k,h}}$.\footnote{Note that this reduction is nondeterministic: it assumes the existence of a prover who can ``guess'' the orders of the two groups $G$ and $\gen{g_1,\ldots,g_k,h}$, which can then be verified using a protocol for Group Order Verification. Such a nondeterministic reduction will be enough since in this paper we only consider complexity classes with a prover.}
The best known upper bound on Group Order Verification is $\AM\cap \coAM$, by Babai \cite{Babai92}.
Since Group Non-Membership, which belongs to $\QMA$, reduces to Group Order Verification,  
this leads to one of the main open problems proposed in \cite{WatrousFOCS00}:

\begin{oproblem}[\cite{WatrousFOCS00}]\label{oprob}
Is Group Order Verification in $\QMA$?
\end{oproblem}
No progress has been made on this problem since 2000.

\subsection{Our results}\label{subsec:results}
\paragraph{Statement of our results.}
In this paper, we prove \cref{conj} and solve \cref{oprob}. Here is our main result:

\begin{theorem}\label{th:order}
Group Order Verification is in $\QCMA$.
\end{theorem}

\cref{th:order} solves \cref{oprob}. Our result is actually significantly stronger: it shows that Group Order Verification is not only in $\QMA$, but also in $\QCMA$. As observed in \cite{Babai92}, an upper bound on the complexity of Group Order Verification leads to the same upper bound for the complement: in order to verify that $\abs{G}\neq m$, Merlin can send the true order of $G$ and then Arthur can use the protocol of \cref{th:order} for checking whether it is really the true order and differs from $m$.
We thus obtain the following stronger statement:

\begin{corollary}\label{cor:order}
Group Order Verification is in $\QCMA\cap\coQCMA$.
\end{corollary}

Since Group Non-Membership reduces to Group Order Verification, 
as another immediate corollary, we obtain a proof of \cref{conj}:

\begin{corollary}\label{th:GNM}
Group Non-Membership is in $\QCMA$.
\end{corollary}

\begin{table}[H]
\centering
\caption{This table compares our new upper bounds with the upper bounds from the literature. 
}\label{table:results}
\renewcommand{\arraystretch}{1.2}
\begin{tabular}{|l|l|l|l|}
    \hline
    \multirow{2}{*}{Problem} & \multicolumn{2}{c|}{Prior upper bounds} & \multicolumn{1}{c|}{New quantum} \\\cline{2-3}
     & \multicolumn{1}{c|}{Classical} & \multicolumn{1}{c|}{Quantum} &\multicolumn{1}{c|}{upper bound} \\\hline\hline
     \multirow{1}{*}{Group Order Verification}& \multirow{1}{*}{$\AM\cap\coAM$ \;\;\;\cite{Babai92}}& - & \multirow{1}{*}{$\QCMA\cap\coQCMA$ (Cor.~\ref{cor:order})}\\\hline
    \multirow{1}{*}{Group Non-Membership} & \multirow{1}{*}{$\AM\cap\coNP$ \cite{Babai92,Babai+FOCS84}}& $\QMA$\;\;\;\;\,\cite{WatrousFOCS00} & \multirow{1}{*}{$\QCMA$ \:\hspace{19mm} (Cor. \ref{th:GNM})}\\
    \hline
    Group Isomorphism & $\AM\cap\Sigma_2^{\mathrm{P}}$ \;\;\;\;\cite{Babai92,Babai+FOCS84} & - & $\QCMA$ \:\hspace{19mm} (Cor. \ref{cor:result1})\\\hline
    Homomorphism, & \multirow{6}{*}{$\AM\cap \coAM$ \;\;\;\cite{Babai92}}& \multirow{2}{*}{-} & \multirow{4}{*}{$\QCMA\cap\coQCMA$ (Cor. \ref{cor:result2})}\\
    Minimal Normal Subg.&&&\\\cline{1-1}\cline{3-3}
    Proper Subgroup && \multirow{1}{*}{$\QMA$\;\;\;\;\,\cite{WatrousFOCS00}} & \\
    \cline{1-1}\cline{3-3}
    Simple Group&  & \multirow{3}{*}{$\coQMA$\;\cite{WatrousFOCS00}} & \\\cline{1-1}\cline{4-4}
    Intersection, Centralizer,&  &  &\multirow{2}{*}{$\coQCMA$ \hspace{17mm}(Cor. \ref{cor:result3})} \\
    Maximal Normal Subg.&&&\\\hline
\end{tabular}\vspace{5mm}
\end{table}

Other than Group Order Verification and Group Non-Membership, we obtain new quantum upper bounds for the complexity of many group-theoretic problems: Group Isomorphism, Homomorphism, Minimal Normal Subgroup, Proper Subgroup, Simple Group, Intersection, Centralizer and Maximal Normal Subgroup (the formal definition of these problems is given in \cref{sec:other}). These eight problems have been considered in the classical setting in \cite{Babai92,Babai+FOCS84}. The last five problems have been considered in the quantum setting in \cite{WatrousFOCS00}. By combining \cref{cor:order} with the proof techniques from \cite{Babai92, WatrousFOCS00}, we easily obtain the following results:

\begin{corollary}\label{cor:result1}
Group Isomorphism is in $\QCMA$.   
\end{corollary}

\begin{corollary}\label{th:homo-ker-MinimalNS}\label{cor:result2}
Homomorphism, Minimal Normal Subgroup, Proper Subgroup and Simple Group are in $\QCMA \cap\coQCMA $.   
\end{corollary}

\begin{corollary}\label{cor:result3}
Intersection, Centralizer and Maximal Normal Subgroup are in  $\coQCMA$.  
\end{corollary}

All the results are summarized in \cref{table:results}.

\paragraph{Related work.}
When writing this paper, we learned from Michael Levet \cite{Levet} and James Wilson~\cite{Wilson} that Alexander Hulpke, Martin Kassabov, {\'A}kos Seress and James Wilson have obtained a proof of the existence of a short presentation for the Ree groups of rank one. The proof, which is 60-page long, is unpublished (and not expected to be published). The existence of such a short presentation leads to an alternative way of proving \cref{prop:Ree}, by using \cref{th:p2} instead of our isomorphism test. 

\subsection{Overview of the proof strategy}
We give below an overview of the strategy we use to prove \cref{th:order}. 

Let us first describe some basic notation and notions of group theory --- more details are given in \cref{sub:gt}. For a group $G$, we write $H\le G$ (resp.,~$H\ns G$) to express that $H$ is a subgroup (resp.,~normal subgroup) of $G$. We denote by $\triv$ the trivial subgroup of $G$. A composition series of a group $G$ is a decomposition of the group into simple groups (a simple group is a nontrivial group that has no nontrivial normal subgroup and thus cannot be further decomposed), which are called the composition factors of $G$. The ``classification theorem of finite simple groups'' states that every finite simple group belongs to one of 18 infinite families of simple groups, or is one of 26 sporadic simple groups. As a consequence, each simple group can be described by a short string called its standard name.

\paragraph{Babai-Beals filtration.}
The starting point of our strategy is the Babai-Beals filtration. Babai and Beals \cite{Babai+99} showed that any group $G$ has a decomposition 
\[
\tgroup\ns\sol{G}\ns\soc{G}\ns\pker{G}\ns G\,,
\]
where $\sol{G}$ and $\pker{G}$ are two normal subgroups of $G$ called the solvable radical and the permutation kernel, respectively, and $\soc{G}$ is another normal subgroup (all these subgroups are defined in \cref{sub:bb}, but their definition is not needed for this overview). 
Ref.~\cite{Babai+99} showed that 
in randomized polynomial time, it is possible to compute a set of generators for $\pker{G}$. Additionally, given a set of generators for $\pker{G}$, it is possible in deterministic polynomial time to test membership in $\pker{G}$ and compute the order $\abs{G/\pker{G}}$. Since \[\abs{G}=\abs{\pker{G}}\cdot\abs{G/\pker{G}},\] in order to compute $\abs{G}$ we thus only need to compute the order of $\pker{G}$.

While how to compute efficiently $\sol{G}$ and $\soc{G}$ is unknown (even with the help of a prover), these two subgroups have an important property: $\sol{G}$ and $\pker{G}/\soc{G}$ are solvable (a solvable group is a group that is ``not too much non-abelian,'' in the sense that all its composition factors are cyclic). Note that while the solvability of $\sol{G}$ is easy to show, the solvability of $\pker{G}/\soc{G}$ is based on Schreier conjecture, which was proposed by Schreier in 1926, and is now known to be true as a result of the classification of finite simple groups (no simpler proof is known). 

\paragraph{Stategy to compute \boldmath{$\abs{\pker{G}}$}.}
In order to compute the order of $\pker{G}$, 
we first observe that the Babai-Beals filtration implies the existence of a solvable subgroup $H_0$ and $2s$ elements $\beta_1,\ldots,\beta_s,\gamma_1,\ldots,\gamma_s\in \pker{G}$ such that, when defining 
\begin{align*}
    H_i&=\gen{H_0,\beta_1,\ldots,\beta_i,\gamma_1,\ldots,\gamma_i}
\end{align*}
for each $i\in[s]$, the chain of inclusions
\begin{equation}\label{eq:chain}
     \tgroup\ns H_0\ns H_1\ns \cdots\ns H_s\ns \pker{G}
\end{equation}
holds, where $\pker{G}/H_s$ is solvable and $H_i/H_{i-1}$ is a simple group
for each $i\in[s]$.\footnote{For instance, we can take $H_0=\sol{G}$ and $H_s=\soc{G}$ to show the existence of such a decomposition. We nevertheless do not require the conditions $H_0 = \sol{G}$ and $H_s= \soc{G}$ since they cannot be checked efficiently (as mentioned above, we do not know how to efficiently compute $\sol{G}$ and $\soc{G}$).} 
Observe that
\[
    \abs{\pker{G}}= \abs{H_0} \cdot 
    \abs{H_1/H_{0}}\cdot 
    \abs{H_2/H_{1}}\cdots 
    \abs{H_s/H_{s-1}}\cdot
    \abs{\pker{G}/{H_s}}
\]
holds. Here $\abs{H_0}$ and $\abs{\pker{G}/H_s}$ are easy to compute with the help of Merlin since they are solvable.\footnote{Actually, $\abs{H_0}$ can be computed even without Merlin's help by using Watrous' algorithm for solvable groups \cite{WatrousSTOC01}.} It thus remains to compute $\abs{H_i/H_{i-
1}}$ for each $i\in[s]$. 

While it is unknown whether the decomposition (\ref{eq:chain}) can be computed in polynomial time, we can ask Merlin to ``guess'' it and send it to Arthur. Concretely, Merlin sends a set of generators for each subgroup $H_i$ and the standard name of the simple group $H_i/H_{i-1}$. The standard names enable Arthur to learn each $\abs{H_i/H_{i-1}}$, and thus to compute $\abs{\pker{G}}$. A dishonest Merlin, however,  might cheat and send a wrong standard name or can even send a series in which some $H_i/H_{i-1}$ is not simple. The main obstacle is thus to check that $H_i/H_{i-1}$ is really isomorphic to the simple group specified by Merlin (by its standard name). 

\paragraph{Isomorphism test.} One promising strategy for testing if a group $\myK$ is isomorphic to a known (not necessarily simple) group~$S$ is to use a randomized homomorphism test. A similar strategy was also used by Aaronson and Kuperberg to analyze the query complexity of Group Non-Membership~\cite{Aaronson+ToC07}.

Let $s_1,\ldots,s_k$ be a set of generators of $S$ known to both Arthur and Merlin. We ask Merlin to send elements $g_1,\ldots,g_k\in 
\myK$. If $\myK\cong S$ and Merlin is honest, he will send $g_i=\phi(s_i)$ for each $i\in[k]$, for some isomorphism $\phi\colon S\to \myK$. For the checking procedure, Arthur defines a map $f\colon S\to \myK$ by extending the partial map $s_i\mapsto g_i$ into a map on all $S$ as if it were a homomorphism. For instance, for an element $s\in S$ that can be written as $s=s_1s_2s_1s_3$, Arthur will set $f(s)=g_1g_2g_1g_3$. 
Arthur then takes two elements $s$ and $s'$ uniformly at random in $S$ and checks if
\begin{equation}\label{eq:hom}
f(ss')=f(s)f(s')
\end{equation}
holds.
By standard results on property testing (e.g., \cite{BenOr+08}), we can show that passing this test with high probability guarantees that there exists a homomorphism from $S$ to $\myK$. 

To be successful, this approach has to satisfy three important requirements:
\begin{itemize}
    \item[A.] Arthur needs to be able to efficiently represent an arbitrary element $s\in S$ as a product of elements from the fixed set $\{s_1,\ldots,s_k\}$. This representation should also be unique for $f$ to be well-defined. 
    \item[B.\,] Arthur needs to be able to efficiently check that the homomorphism whose existence is guaranteed when passing the homomorphism test is actually an isomorphism, i.e., a bijection.
    \item[C.] Arthur needs to be able to efficiently check if (\ref{eq:hom}) holds. 
\end{itemize}

The first two requirements were also mentioned by Aaronson and Kuperberg \cite{Aaronson+ToC07} as obstacles to prove that Group Non-Membership is in $\QCMA$. In particular, Task B was handled in \cite{Aaronson+ToC07} by solving an instance of the Normal Hidden Subgroup Problem using the quantum algorithm by Ettinger, H{\o}yer and Knill \cite{Ettinger+04}, which has polynomial query complexity but in general exponential time complexity.

Note that Aaronson and Kuperberg \cite{Aaronson+ToC07} applied the homomorphism test to the whole group $\myK=G$. In our strategy, however, we are working on a composition factor $\myK=H_i/H_{i-1}$, i.e., we only need to consider the case where $S$ is a simple group. For simple groups, Tasks A and B can be implemented efficiently. For Task A, simple groups have a concrete representation for which we can efficiently represent any element as a product of elements from the fixed set (this is nontrivial and requires advanced techniques, such as the machinery for matrix groups developed by Babai, Beals and Seress \cite{Babai+STOC09}). For Task B, we can fairly easily guarantee that the homomorphism is a bijection by exploiting the property that simple groups do not have nontrivial subgroups (another interpretation is that the Normal Hidden Subgroup Problem is easy in simple groups since the only normal subgroups are the trivial subgroup and the whole group). In other words, we are able to bypass the first two obstacles because we are working on composition factors, and not on the whole group.

The price to pay is that Task C now becomes very challenging. When working on the whole group $\myK=G$ as done in \cite{Aaronson+ToC07}, checking if (\ref{eq:hom}) holds is trivial since the oracle for the black-box group $\bbgroup$ can be directly applied. When $\myK=H_i/H_{i-1}$, this is not the case anymore: checking if (\ref{eq:hom}) holds is equivalent to checking if 
\[
f(ss')f(s')^{-1}f(s)^{-1}\in H_{i-1}
\]
holds, which requires the ability to check membership in $H_{i-1}$.  This is challenging since the group $H_{i-1}$ can be arbitrary. Note that it is crucial that the elements $s$ and $s'$ chosen by Arthur are hidden from Merlin, otherwise Merlin can cheat by choosing $g_i$'s such that \cref{eq:hom} holds (only) for those specific $s$ and $s'$. For this reason, we cannot use Merlin to directly help Arthur check membership in $H_{i-1}$ (e.g., by sending a membership certificate). Instead, Arthur should be able to efficiently test membership in $H_{i-1}$ by himself. This is the main difficulty we have to overcome in this work.  

\paragraph{Replacing membership in \boldmath{$H_{i-1}$} by membership in \boldmath{$H_0$}.}
We show (in \cref{th:group}) the following crucial consequence of the Babai-Beals filtration: there exists a decomposition of the form (\ref{eq:chain}) that satisfies the  additional condition
\[
    H_i/H_{i-1}\cong \gen{H_0,\beta_i,\gamma_i}/H_0\,\textrm{, for all }i\in[s]. \tag{$\star$}
\]

In our protocol for Group Order Verification (described in \cref{sec:main}), an honest Merlin sends a decomposition satisfying $(\star)$. In order to check if $H_i/H_{i-1}\cong S_i$ for some specific simple group $S_i$, it is thus enough to check if $\gen{H_0,\beta_i,\gamma_i}/H_0\cong S_i$. As explained above, to use the homomorphism test, we need to be able to efficiently check membership in $H_0$. Crucially, $H_0$ is now a solvable group, and we can thus use Watrous' polynomial-time quantum algorithm for membership testing in solvable groups \cite{WatrousSTOC01} to implement the homomorphism test efficiently. 

In the case of a dishonest Merlin, we can still guarantee that the decomposition (\ref{eq:chain}) satisfies $\gen{H_0,\beta_i,\gamma_i}/H_0\cong S_i$ for all $i\in[s]$,\footnote{A decomposition such that $\gen{H_0,\beta_i,\gamma_i}/H_0\cong S_i$ for all $i\in[s]$ is called a \emph{nice decomposition} in \cref{sec:bb} (see \cref{def:nice}).} but cannot guarantee that it satisfies Condition ($\star$). We are nevertheless able to show (in \cref{th:gmain}) that $\gen{H_0,\beta_i,\gamma_i}/H_0\cong S_i$ implies that $\abs{H_i/H_{i-1}}$ is a divisor of $\abs{S_i}$, and show that guaranteeing that $\abs{H_i/H_{i-1}}$ is a divisor of $\abs{S_i}$ is enough for our purpose. 

In order to finish establishing the soundness of the protocol, further work is needed. We should especially deal with the potential cheating strategy in which Merlin sends, instead of (\ref{eq:chain}), the chain of subgroup
\begin{equation}
     \tgroup\ns H_0\ns H_1\ns \cdots\ns H_s\ns K
\end{equation}
for a proper subgroup $K\lneq\pker{G}$.
To prevent such a cheating, we observe that $K=\pker{G}$ if and only if the composition factors of $G/\pker{G}$ match (with multiplicity) the composition factors of $G/K$. We then use another deep property of the Babai-Beals filtration: $G/\pker{G}$ is isomorphic to a symmetric group of small degree, which implies that the composition factors of $G/\pker{G}$ are fairly ``easy''. We can thus check if the composition factors of $G/\pker{G}$ match (with multiplicity) the composition factors of $G/K$ fairly easily with the help of Merlin.

\paragraph{The Ree groups of rank one.}
Instead of testing if $\gen{H_0,\beta_i,\gamma_i}/H_0\cong S_i$ for each $i\in[s]$ using the isomorphism test described above, we observe that we actually only need to do it for one class of simple groups: the Ree groups of rank one. For all the other simple groups, we can use a simpler approach, already proposed in \cite{Babai+99}, based on the existence of short presentations (see \cref{th:p2}). In this paper, we thus describe the isomorphism test only for the Ree groups of rank one (in \cref{sec:Ree}).

\section{Preliminaries}
In this section, we describe the notions of complexity theory, group theory and black-box groups needed to show our results. For any positive integer $s$, we write $[s]=\{1,\ldots,s\}$.
\subsection{Quantum complexity theory}
We assume that the reader is familiar with the most basic concepts and terminology of quantum computing, such as quantum circuits and measurements. The main technical contribution of this work is to construct \emph{classical} certificates for order verification that can be checked by known quantum algorithms (e.g., Watrous' quantum algorithms for solvable groups \cite{WatrousSTOC01}). 
The claims of this paper can be verified without further expertise in quantum computing if the reader is willing to consider these quantum algorithms as black-boxes. 
We just give below the formal definition of the complexity class $\QCMA$.

\begin{definition}\label{def:QCMA}
A problem $A=(A_{{\rm yes}},A_{{\rm no}})$ is in $\QCMA$ if there is a polynomial-time quantum algorithm~$V$ (by a verifier called Arthur) such that:
\begin{description}
    \item[Completeness] For any $x\in A_{{\rm yes}}$, there is a polynomial-length binary string $w_x$ called a certificate (from a prover called Merlin) such that $V$ accepts on input $(x,w_x)$ with probability at least $2/3$.
    \item[Soundness] For any $x\in A_{{\rm no}}$, $V$ accepts with probability at most $1/3$ on input $(x,w)$ for any polynomial-length binary string $w$.
\end{description}
The quantum algorithm $V$ with certificates $\{w_x\}_{x\in A_{{\rm yes}}}$ is called a $\QCMA$ protocol.
\end{definition}

\subsection{Group theory}\label{sub:gt}
A group $G$ is called solvable if there exist $g_1,\cdots,g_s\in G$ such that when defining $H_i=\gen{g_1,\ldots,g_i}$ for each $i\in[s]$, 
\[
\{e\}=H_0\ns H_1\ns \cdots \ns H_s=G.
\]
Note that $H_{i}/H_{i-1}$ is necessarily cyclic in this case, for each $i\in[s]$. For any integer $k\ge 1$, we denote by $\Sym(k)$ the symmetric group of degree $k$. Let $S$ be a set of generators of the group $G$. We call a sequence $(g_1,\ldots,g_t)$ of elements of $G$ a \emph{straight-line program} over $S$ if each $g_i$ is either a member of $S$ or an element of the form $g_j^{-1}$ or $g_jg_k$ from some $j,k<i$. The length of the straight-line program is $t$. The element reached by the straight-line program is the last element $g_t$. 

We will use the following easy fact in \cref{sec:main}. 

\begin{fact}(see, e.g., \cite[Proposition 5.42]{Rotman'02})\label{lemma:nilpotent}
For any finite group $G$ and any prime power $p^t$, $p^t$ divides $|G|$ if and only if $G$ has a subgroup of order $p^t$.
\end{fact}

\paragraph{Composition series and simple groups.}
A simple group is a nontrivial group with no nontrivial normal subgroup. Any group can be decomposed into simple groups via composition series.
\begin{definition}\label{def1}
Let $G$ be a finite group. A composition series of $G$ is a list of subgroups
$H_0,H_1,\ldots,H_s$, for some integer $s$, such that
\begin{itemize}
\item[(a)]
$\{e\}=H_0\ns H_1\ns \cdots \ns H_s=G$\,;
\item[(b)]
the quotient group $H_{i}/H_{i-1}$ is a simple group for each $i\in[s]$. 
\end{itemize}
The composition factors of $G$ are the quotients $H_1/H_0,H_2/H_1,\ldots,H_s/H_{s-1}$.
\end{definition}
Any finite group has a composition series (with $s=O(\log |G|)$).  While a group may have more than one composition series, the Jordan-H\"older theorem (e.g., \cite[Theorem 22]{Dummit+04}) shows that they have the same length and the same composition factors, up to permutation and isomorphism. 

The ``classification theorem of finite simple groups'' (see, e.g., \cite{Wilson09}) is a theorem which states that every finite simple group belongs to one of 18 infinite families of simple groups (each family being indexed by one or two parameters), or is one of 26 sporadic simple groups. This gives $18+26=44$ types of finite simple groups. 
Each finite simple group $S$ can thus be represented by a binary string $z=(z_1,z_2)$ where $z_1$ is a constant-length binary string indicating its type
and $z_2$ is a $O(\log\abs{S})$-length string representing its parameters ($z_2$ is empty if $S$ is a sporadic simple group).   
We call $z$ the \emph{standard name} of the finite simple group $S$. Conversely, given a binary string $z$ corresponding to a standard name of a finite simple group, we write $\group{z}$ the simple group represented by $z$. 
Given a standard name $z$, the order $\abs{\group{z}}$ can be easily computed. 

\paragraph{The Ree groups of rank one.}\label{sub:Ree}
The family of Ree groups of rank one is among the least understood families of finite simple groups. In this paper, we define these groups by using their natural matrix representation, and use the same set of generators as in \cite{Baarnhielm14,Kemper+01}.

The family of Ree groups of rank one is indexed by a positive integer $a$. 
Write $q=3^{2a+1}$ and $t=3^a$. Let $\field_q$ be the finite field of order $q$ and $\omega$ be a primitive element of $\field_q$.
Consider the group $\GL(7,q)$ of invertible matrices of dimension 7 over $\field_q$. The Ree group of rank one, which we denote by $\Ree$,\footnote{This group is also written as $\oRee{q}$ in the literature.} is the subgroup of $\GL(7,q)$ generated by the following three matrices:

\begin{align*}
\Gamma_1 &= 
\begin{bmatrix}
1 & 1 & 0 & 0 & -1 & -1 & 1 \\
0 & 1 & 1 & 1 & -1 & 0 & -1 \\
0 & 0 & 1 & 1& -1 & 0 & 1 \\
0 & 0 & 0 & 1 & 1 & 0 & 0 \\
0 & 0 & 0 & 0 & 1 & -1 & 1 \\
0 & 0 & 0 & 0 & 0 & 1 & -1 \\
0 & 0 & 0 & 0 & 0 & 0 & 1
\end{bmatrix},\hspace{25mm}
\Gamma_2 = 
\begin{bmatrix}
0 & 0 & 0 & 0 & 0 & 0 & -1 \\
0 & 0 & 0 & 0 & 0 & -1 & 0 \\
0 & 0 & 0 & 0 & -1 & 0 & 0 \\
0 & 0 & 0 & -1 & 0 & 0 & 0 \\
0 & 0 & -1 & 0 & 0 & 0 & 0 \\
0 & -1 & 0 & 0 & 0 & 0 & 0 \\
-1 & 0 & 0 & 0 & 0 & 0 & 0
\end{bmatrix},\\
\Gamma_3&=
\begin{bmatrix}
    \omega^t & 0 & 0 & 0 & 0 & 0 & 0 \\
    0 & \omega^{1-t} & 0 & 0 & 0 & 0 & 0 \\
    0 & 0 & \omega^{2t-1} & 0 & 0 & 0 & 0 \\
    0 & 0 & 0 & 1 & 0 & 0 & 0 \\
    0 & 0 & 0 & 0 & \omega^{1-2t} & 0 & 0 \\
    0 & 0 & 0 & 0 & 0 & \omega^{t-1} & 0 \\
    0 & 0 & 0 & 0 & 0 & 0 & \omega^{-t}
    \end{bmatrix}. 
\end{align*}

The order of $\Ree$ is $q^3(q^3+1)(q-1)$. 

\paragraph{Short presentations of groups.}
A presentation of a group $G$ is a definition of $G$ in terms of generators and relations (see, e.g., \cite{Dummit+04}). The length of the presentation is the total number of characters required to write down all relations between the generators. It is known that most simple groups have short presentations:

\begin{theorem}[\cite{Babai+97,Hulpke+01}]\label{th:p1}
    All the finite simple groups, with the possible exception of the family of Ree groups of rank one, have a polylogarithmic-length presentation (i.e., a presentation of length polynomial in the logarithm of the order of the group). Moreover, these polylogarithmic-length presentations can be efficiently computed from the standard name of the simple group.
\end{theorem}

In particular, any composition factor of a solvable group has a polylogarithmic-length presentation. For any group isomorphic to a subgroup of a permutation group of small degree, any composition factor also has a polylogarithmic-length presentation. This result is well-established in the literature (see, e.g., \cite[Section 6]{Luks1997}), but not explicitly stated. We provide a proof here for completeness.
\begin{theorem}\label{th:smallk}
If $G\le \Sym(k)$, then any composition factor of $G$ has a $\poly(k)$-length presentation.
\end{theorem}
\begin{proof}

Let $\tgroup= H_0\ns H_1\ns \cdots\ns H_s\le G$ be a composition series of $G$. Each $H_{i-1}$ is a maximal normal subgroup in $H_{i}$ and $H_{i}/H_{i-1}$ is a simple group. Then for each $i \in [s]$, $H_{i}/H_{i-1}$ is isomorphic to a subgroup of $\Sym(k)$ (see, e.g, \cite{Derek-Holt-MOF}).




By \Cref{th:p1}, $H_{i}/H_{i-1}$ has a presentation of length $\poly(k)$ except when $H_{i}/H_{i-1} \cong \Ree$ is a Ree group of rank one. 

Suppose $H_{i}/H_{i-1} \cong \Ree$. It is known that $(q^3+1)$ is the least positive integer such that there is an injective homomorphism $\phi:\Ree \to \Sym(q^3+1)$ (see, e.g., \cite[Table 4]{GMPS+2015}). Therefore, $|H_{i}/H_{i-1}|=|\Ree|=q^3(q^3+1)(q-1) \leq (q^3+1)^{3}  \leq k^3$. 
Taking the generating set as the set of all elements of $H_{i}/H_{i-1}$ and the set of relations as all the multiplication relations in the multiplication table, we get a $\poly(k)$-length presentation for $H_{i}/H_{i-1} \cong \Ree$. 
\end{proof}

\paragraph{Homomorphism test.} Testing the homomorphism of finite groups is a well-studied topic in property testing. In this paper we will use the following result from \cite{Aaronson+ToC07}, which is based on~\cite{BenOr+08}.
\begin{lemma}[Propositions 5.2 and 5.3 in \cite{Aaronson+ToC07}]\label{lemma:homo}
    Let $G,G'$ be two groups and consider a function $f\colon G\to G'$.
    Assume that the inequality
    \begin{equation}\label{eq:ineq}
    \Pr_{r_1,r_2\in G}\:[f(r_1r_2)=f(r_1)f(r_2)]\ge 9/10
    \end{equation}
    holds.
    Then there exists a unique homomorphism
    $\phi\colon G\to G'$ such that
    \[
        \Pr_{x\in G}\:[f(x)\neq\phi(x)]\le 1/10.
    \]
\end{lemma}

\subsection{Black-box groups}\label{sub:black}
A black-box group is a representation of a group $\bbgroup$ introduced by Babai and Szemer\'{e}di~\cite{Babai+FOCS84} in which each element of $\bbgroup$ is encoded by a binary string of a fixed length $n=O(\log |\bbgroup|)$. 
Let $s\colon \bbgroup\to \{0,1\}^n$ denote the encoding of elements as binary strings. If $s$ is injective, we say that $\bbgroup$ is a black-box group with unique encoding. If the encoding is not unique, an oracle for identity testing (i.e., deciding whether or not a given string encodes the identity element of $\bbgroup$) is available.

In the classical setting, classical oracles are available to perform group operations (each call to the oracles can be done at unit cost). A first oracle performs the group product: given two strings representing two group elements $g$ and $h$, the oracle outputs the string representing $gh$. A second oracle performs inversion: given a string representing an element $g\in \bbgroup$, the oracle outputs the string representing the element $g^{-1}$. These two oracles output an error message on non-valid inputs (i.e., strings in $\{0,1\}^n\setminus s(\bbgroup)$).
All the classical algorithms and protocols discussed in this paper do not require that $\bbgroup$ has unique encoding.

In the quantum setting, the oracles performing the group operations have to be able to deal with quantum superpositions. 
As in prior works \cite{Aaronson+ToC07,Ivanyos+03,LeGall+MFCS18,LeGallSTACS10,WatrousFOCS00,WatrousSTOC01}, in the quantum setting we always consider black-box groups with unique encoding. Let $s\colon \bbgroup\to \{0,1\}^n$ denote the injective encoding of elements as binary strings. Two quantum oracles are available (each call to the oracles can be done at unit cost). The first oracle maps $\ket{s(g)}\ket{s(h)}$ to $\ket{s(g)}\ket{s(gh)}$, for any two elements $g,h\in \bbgroup$. The second quantum oracle maps $\ket{s(g)}\ket{s(h)}$ to $\ket{s(g)}\ket{s(g^{-1}h)}$, for any $g, h\in \bbgroup$. These two oracles output an error message on non-valid inputs (in the quantum setting this is implemented by introducing a third 1-qubit register that is flipped when the inputs are not valid --- see \cite{WatrousFOCS00} for details).

We describe below several classical and quantum techniques for black-box groups.

\paragraph{Testing solvability and approximate sampling in black-box groups.} We will use the following classical randomized algorithms from \cite{BabaiSTOC91,Babai+99} to sample nearly uniformly elements in black-box groups and test solvability. 
\begin{theorem}[\cite{BabaiSTOC91}]\label{th_babai}
    Let $\bbgroup$ be a black-box group. 
    For any $G\le \bbgroup$ and any $\varepsilon>0$,
    there exists a classical randomized algorithm running
    in time polynomial in $\log(|\bbgroup|)$ and $\log(1/\varepsilon)$ that outputs an element of $G$ such that each $g\in G$ is output with probability in range $(1/|G|-\varepsilon,1/|G|+\varepsilon)$. 
    \end{theorem}
\begin{theorem}[\cite{Babai+99}]\label{lemma:sol}
    Let $\bbgroup$ be a black-box group. 
    For any $G\le \bbgroup$, there exists a classical randomized algorithm running
    in time polynomial in $\log(|\bbgroup|)$ that decides if $G$ is solvable.
\end{theorem}

\paragraph{Watrous' algorithms for solvable groups.}
Watrous showed that for solvable groups, Group Order Verification and Group Non-Membership can be solved by polynomial-time quantum algorithms. Here are the precise statements we will use in our paper:

\begin{theorem}{(\cite{WatrousSTOC01})}\label{prop:Watrous}
Let $\bbgroup$ be a black-box group. 
There exist quantum algorithms running in time $\poly(\log\abs{\bbgroup})$ that solve the following tasks with probability at least $1-1/\poly(|\bbgroup|)$:
\begin{itemize}
    \item given a solvable group $G\le \bbgroup$, compute $|G|$;
    \item given a solvable group $G\le \bbgroup$ and an element $g\in \bbgroup$, decide if $g\in G$.
\end{itemize}
\end{theorem}
We immediately obtain the following corollary.
\begin{corollary}\label{cor:Watrous}
    Let $\bbgroup$ be a black-box group. 
    There exists a quantum algorithm running in time $\poly(\log\abs{\bbgroup})$ that given a group $G\le \bbgroup$ and a solvable subgroup $H\le G$, decide if $H$ is normal in $G$ with probability at least $1-1/\poly(\abs{\bbgroup})$.
\end{corollary}

\paragraph{Membership, normality, solvability and isomorphism certificates.}
Babai and Szemer\'edi \cite{Babai+FOCS84} showed that for any group $G$, any set of generators $S\subseteq G$ and any element $g\in G$, there exists a straight-line program over $S$ of length at most $(1+\log \abs{G})^2$ that generates $g$. This leads to the concept of \emph{membership certificate}. 
\begin{definition}\label{def:mem-cert}
    Let $\bbgroup$ be a black-box group. 
    For a group $G\le \bbgroup$ given by a set of generators $S$ and an element $g\in G$, a certificate of membership of $g$ in $G$ is a straight-line program over $S$ of length at most $(1+\log \abs{\bbgroup})^2$ that generates $g$.
\end{definition}
Since such a membership certificate can be verified in polynomial time using the group oracle, the existence of membership certificate established in \cite{Babai+FOCS84} shows that Group Membership is in the class $\NP$. Next, we introduce the notion of \emph{normality certificate}, which was also used in \cite{Babai+FOCS84} for checking the normality of a subgroup.
\begin{definition}\label{def:norm-cert}
    Let $\bbgroup$ be a black-box group. 
    For any group $G\le \bbgroup$ given as $G=\gen{g_1,\ldots,g_t}$ and any subgroup $H\le G$ given as $H=\gen{h_1,\ldots,h_s}$, a normality certificate of $H$ in $G$ is a collection of membership certificates for the inclusions 
    \begin{equation}\label{eq:norm}
        g_jh_ig_j^{-1}\in H, \textrm{ for each } i\in[s] \textrm{ and each } j\in[t].
    \end{equation}    
\end{definition}

We now introduce the notion of solvability certificate \cite{Babai+FOCS84}, which is used for checking if a group is a solvable group of order dividing a given integer.
\begin{definition}\label{def:sol}
    Let $\bbgroup$ be a black-box group. 
    For a group $G\le\bbgroup$ and an integer $m$, a
    certificate that $G$ is a solvable group of order dividing $m$ is  
    \begin{itemize}
        \item a list of $s$ primes $(p_1,\ldots,p_s)$ such that $p_1\cdots p_s=m$\,;
        \item a set of elements $g_1,\ldots,g_s\in G$ along with certificates certifying membership in $G$\,;
        \item for each $i\in[s]$, a normality certificate certifying that $\gen{g_1,\ldots,g_{i-1}}\ns \gen{g_1,\ldots,g_i}$\,;
        \item for each $i\in [s]$, a certificate of the inclusion $g_i^{p_i}\in \gen{g_1,\ldots,g_{i-1}}$\,.
    \end{itemize}
\end{definition}

Babai and Szemer\'edi have shown (see \cite[Theorem 11.4]{Babai+FOCS84}) that the order of a black-box group is certifiable if all its composition factors are isomorphic to simple groups that have  polylogarithmic-length presentations.
We will  need the following slightly different statement, which is proved by a similar technique. 
\begin{proposition}\label{th:p2}
    Let $\bbgroup$ be a black-box group 
    and
    $S$ be a simple group that has a polylogarithmic-length presentation.
    For any $G\le \bbgroup$,
    the problem of testing if $G\cong S$
    is in $\NP$.
\end{proposition}
\begin{proof}
    Let $\alpha_1,\ldots\alpha_s$ be the generators of $S$, and $\Pp$ be the set of relations between those generators in the polylogarithmic-length presentation of $S$. Let $\{g_1,\ldots,g_t\}$ denote the set of generators of $G$. If $G=\tgroup$ (which can be easily checked) we know that $G$ is not isomorphic to $S$. We thus assume below that $G\neq\tgroup$.

    The prover sends elements $g'_1,\ldots,g'_s$ of $G$, along with certificates of the equality $\gen{g'_1,\ldots,g'_s}=\gen{g_1,\ldots,g_t}$. The verifier checks if the certificates are correct (which guarantees that $\gen{g'_1,\ldots,g'_s}=G$) and checks if each relation in $\Pp$ holds in $G$ when replacing $\alpha_i$ by $g'_i$ for all $i\in[s]$ (which guarantees\footnote{This is a folklore fact (see for instance \cite[Section 6.3]{Dummit+04}). Here is a proof: the map $\varphi(\alpha_i)=g'_i$ can be extended into a surjective homomorphism $\varphi\colon S\to\gen{g'_1,\ldots,g_s'}$. We thus have $\gen{g'_1,\ldots,g_s'}\cong S/\ker{(\varphi)}$.} that the group $\gen{g'_1,\ldots,g_s'}$ is isomorphic to $S/N$ for some normal subgroup $N$ of $S$). 
    
    If $G$ is isomorphic to $S$, then there exists an isomorphism $\phi\colon S\to G$. The prover sends the elements $g'_1,\ldots,g'_s$ of $G$ such that $\phi(\alpha_i)= g'_i$ for each $i\in[s]$. We have $\gen{g'_1,\ldots,g'_s}=\gen{g_1,\ldots,g_t}$, and thus the prover can also send correct certificates of this equality. Since $\phi$ is a homomorphism, each relation in $\Pp$ holds in $G$ when replacing $\alpha_i$ by $g'_i$ for all $i\in[s]$. Thus all the tests succeed.

    Conversely, if all the tests succeed we know that $G=\gen{g'_1,\ldots,g'_s}$ is isomorphic to $S/N$ for some normal subgroup $N$ of $S$. Since $S$ is simple, its only normal subgroups are $\{e\}$ or $S$. Since the case $G=\tgroup$ is excluded, we conclude that $G$ is isomorphic to $S$. 
\end{proof}
Using \cref{th:p2}, we show the following result.
\begin{proposition}\label{th:p3}
    Let $\bbgroup$ be a black-box group 
    and
    $\mathcal{S}$ be a multiset of
    simple groups that each has a polyloga{-}rithmic-length presentation and is given by its standard name.
    For any $G\le \bbgroup$,
    the problem of testing if the multiset of composition factors of $G$ is $\mathcal{S}$ is in $\NP$.
\end{proposition}
\begin{proof}
Merlin guesses a composition series of $G$\,:
\[
\{e\}=H_0\ns H_1\ns \cdots \ns H_s=G\,.
\]
For each $i\in[s]$, Merlin sends to Arthur a set of generators for $H_i$, a normality certificate certifying that $H_{i-1}$ is normal in $H_i$ and the standard name $z_i$ of the simple group $H_i/H_{i-1}$. Arthur checks that the normality certificates are correct and also checks that $\{z_1,\ldots,z_s\}$ is equal to the multiset of standard names corresponding to ${\cal S}$. He then checks that $H_i/H_{i-1}\cong\group{z_i}$ for all $i\in [s]$ using the protocol from \cref{th:p2}. 
\end{proof}

\section{The Babai-Beals Filtration and Nice Group Decompositions}\label{sec:bb}
In this section, we define our concept of nice decomposition of a group, inspired by the Babai-Beals filtration. We first describe the Babai-Beals filtration in \cref{sub:bb}. We then introduce the notion of nice decomposition and show several important properties in \cref{sub:nice}.

\subsection{The Babai-Beals filtration}\label{sub:bb}
The following characteristic chain of subgroups introduced by Babai and Beals \cite{Babai+99} has become a fundamental tool in the algorithmic theory of matrix and black-box groups: for any group $G$,
\[
\tgroup\le\sol{G}\le\soc{G}\le\pker{G}\le G\,.
\]
We now define each of the terms in this chain. 
\begin{itemize}
\item
$\sol{G}$, the solvable radical of $G$, is the unique largest solvable normal subgroup of $G$. 
\item
${\rm Soc}(G)$, the socle of $G$, is the subgroup generated by all minimal normal subgroups of $G$. 

\item
$\soc{G}$ is the preimage of ${\rm Soc}(G/\sol{G})$ in the natural projection $G \longrightarrow G/\sol{G}$, i.e., $\soc{G}/\sol{G}={\rm Soc}(G/\sol{G})$. The group $\soc{G}/\sol{G}={\rm Soc}(G/\sol{G})$ is the direct product of simple groups $T_1, \ldots, T_k$. The group $G$ acts on the set $\{T_1, \ldots, T_k\}$ via conjugational action. Let $\phi: G \longrightarrow \Sym(k)$ denote the permutation representation of $G$ via conjugation action on the set $\{T_1, \ldots, T_k\}$. 
\item
$\pker{G}$, the permutation kernel of $G$, is the kernel of $\phi$, i.e.,  $\pker{G}=\ker(\phi)$.
\end{itemize}
In this paper, we will not require detailed knowledge of all the algebraic structure of these terms. Instead, we will only need the following four properties (\cite[Section 1.2]{Babai+99}) : 

\begin{itemize}
    \item[(a)] $\sol{G}$ is solvable;
    \item[(b)] $\soc{G}/\sol{G}$ is a direct product of simple groups; 
    \item[(c)] $\pker{G}/\soc{G}$ is solvable;
    \item[(d)] $G/\pker{G} \leq \Sym(k)$ and $k \leq \frac{\log |G|}{\log 60}$.
\end{itemize}

A set of generators of $\pker{G}$ can be computed in Monte Carlo polynomial time \cite[Theorem~1.1]{Babai+99}.\footnote{Theorem 1.1 in \cite{Babai+99} assumes that a superset of the set of primes dividing $\abs{G}$ is given as an additional input. These primes are used to compute the order of elements of~$G$ efficiently. In the quantum setting, this is not needed since the order of an element can be computed in polynomial time~\cite{Cheung+01,Shor97}.}  
Moreover, in Monte Carlo polynomial time, we can set up a data structure which allows, for any $g \in G$, to compute $\phi(g)$ in deterministic polynomial time (\cite[Corollary 5.2]{Babai+99}). Using this data structure, membership in $\pker{G}$ can checked in deterministic polynomial time (see the discussion after Theorem 1.1 in \cite{Babai+99}). The following tasks can also be solved in deterministic polynomial time:\footnote{As explained after Theorem 1.1 in \cite{Babai+99}, this can be done by 
    using the extensive library of polynomial-time algorithms for permutation groups \cite{BaLS,Dixon+96,Luks1997,Luks1993}. More specifically, we can use the algorithms from \cite[Section 6]{Luks1997}, \cite[Section 3]{Luks1993} or \cite[Section 3.6]{Dixon+96}.}

\begin{itemize}
    \item compute a set of generators for $G/\pker{G} \leq \Sym(k)$ as permutations in $\Sym(k)$;
    \item compute the multiset of composition factors of $G/\pker{G}$, where each composition factor is given by its standard name, and thus compute $|G/\pker{G}|$.
\end{itemize}

\subsection{Nice decompositions and their properties}\label{sub:nice}
We now introduce our concept of nice decomposition.
\begin{definition}\label{def:nice}
    Let $P$ be a group. A nice decomposition of $P$ consists of a solvable normal subgroup $H_0\ns P$ and $2s$ elements $\beta_1,\ldots,\beta_s,\gamma_1,\ldots,\gamma_s\in P$ such that, when defining 
    \begin{align*}
        H_i&=\gen{H_0,\beta_1,\ldots,\beta_i,\gamma_1,\ldots,\gamma_i}
    \end{align*}
    for each $i\in[s]$, the following conditions are satisfied:
    \begin{itemize}
        \item[(C1)] $\tgroup\le H_0\ns H_1\ns \cdots\ns H_s\ns P$,
        \item[(C2)]  $P/H_s$ is solvable,
        \item[(C3)] $\gen{H_0,\beta_i,\gamma_i}/H_0$ is simple, for all $i\in[s]$.
    \end{itemize}
\end{definition}

We first state a general property of nice decompositions that will be used in \cref{sec:main} to analyze the soundness of our protocol.

\begin{proposition}\label{th:gmain}
    For any group $P$ and any nice decomposition of $P$ (with the same notations as in \cref{def:nice}), 
    $\abs{H_i/H_{i-1}}$ divides $\abs{\gen{H_0,\beta_i,\gamma_i}/H_0}$ for any $i\in[s]$.
\end{proposition}

\begin{proof}
Define $H'=\gen{H_0,\beta_i,\gamma_i}\cap H_{i-1}$, which is a normal subgroup of $\gen{H_0,\beta_i,\gamma_i}$ since $H_{i-1}\ns H_i$.
Observe that any element $g\in H_i$ can be written as 
\begin{equation}\label{eq:d}
    g=h_0\,w_1(g)\cdots w_{i-1}(g)\,w_i(g),
\end{equation}
where $h_0\in H_0$ and $w_\ell(g)$ means one word over the alphabet $\{\beta_\ell,\gamma_\ell\}$. This representation is not unique, but if~$g$ also admits the representation 
\[
    g=h'_0\,w'_1(g)\cdots w'_{i-1}(g)\,w'_i(g),
\]
we have $w_i(g){w'_i}^{-1}(g)\in \gen{\beta_i,\gamma_i}\cap H_{i-1} \le H'$. 

Define the map
$\phi:H_i \longrightarrow \frac{\gen{H_0,\beta_i,\gamma_i}}{H'}$ using Eq.~(\ref{eq:d}) as
\[
    \phi\big(h_0\,w_1(g)\cdots w_{i-1}(g)\,w_i(g)\big)=w_i(g)H'
\]
(the above observation guarantees that the map is well-defined).
This is a homomorphism: for any elements $g,g'\in H_i$ written as 
\begin{align*}
    g&= h_0\,w_1(g)\cdots w_{i-1}(g)\,w_i(g),\\
    g'&= h'_0\,w'_1(g)\cdots w'_{i-1}(g)\,w'_i(g),
\end{align*}
we obtain 
\[
\phi(gg')=w_i(g)w_i(g')H'=\phi(g)\phi(g').
\]
We have 
\[
    \ker(\phi)=\left\{h_0\,w_1(g)\cdots w_{i-1}(g)\,w_i(g)\:|\:w_i(g)\in H'\right\}=H_{i-1}.
\]
Clearly, $\phi$ is surjective. 
By the first homomorphism theorem of groups, we conclude that \[\frac{H_i}{{H_{i-1}}}\cong \frac{\gen{H_0,\beta_i,\gamma_i}}{H'}.\]
Since $H_0\le H'$, we conclude that $\abs{H_i/H_{i-1}}$ divides $\abs{\gen{H_0,\beta_i,\gamma_i}/H_0}$.
\end{proof}

By using the Babai-Beals filtration we can prove the following theorem, which shows that $\pker{G}$ has a nice decomposition with a very useful additional property (Property ($\star$)). This property will be crucial in \cref{sec:main} to establish the completeness of our protocol.

\begin{theorem}\label{th:group}
    For any group $G$, the subgroup $\pker{G}$ has a nice decomposition satisfying the condition
    \[
        H_i/H_{i-1}\cong \gen{H_0,\beta_i,\gamma_i}/H_0\,\textrm{, for all }i\in[s]. \tag{$\star$}
    \]
\end{theorem}
\begin{proof}
We use the Babai-Beals filtration. From the discussion in \cref{sub:bb}, there exists an isomorphism 
\[
    \phi\colon T_1 \times \cdots \times T_k \to \soc{G}/\sol{G},
\]
for simple groups $T_1,\ldots,T_k$. Since every simple group admits a generating set of size two (see, e.g., \cite{guralnick+00}), for each $i\in[k]$ we can write $T_i=\gen{a_i,b_i}$. We take $s=k$ and $H_0=\sol{G}$. For any $i\in[s]$, we take (arbitrary) elements $\beta_i,\gamma_i\in \soc{G}$ such that 
\begin{align*}
    \phi((e,\ldots,e,a_i,e\ldots,e))&=\beta_i H_0\,,\\
     \phi((e,\ldots,e,b_i,e\ldots,e))&=\gamma_i H_0\,,
\end{align*}
where in the above equations $a_i$ and $b_i$ are at the $i$-th position. For each $i\in[s]$,
define the subgroup $H_i=\gen{H_0,\beta_1,\ldots,\beta_i,\gamma_1,\ldots,\gamma_i}$. 
Note that $H_s=\soc{G}$ and $H_{i-1}\ns H_i$ for any $i\in[s]$. Additionally, $\pker{G}/H_s=\pker{G}/\soc{G}$ is solvable, as explained in \cref{sub:bb}. Conditions (C1) and~(C2) of Definition \ref{def:nice} thus hold for $P=\pker{G}$. Since $H_0 \ns G$, we also have $H_0 \ns \pker{G}$, as required for a nice decomposition.

Note that the following property holds: 
\begin{equation}\label{eq:conj}
    \gen{H_0,\beta_i,\gamma_i} \cap H_{i-1}  = H_0\:\textrm{ for any } i\in[s].
\end{equation}
Observe that any element $g\in H_i$ can be written as 
\[
    g=h_0\,w_1(g)\cdots w_{i-1}(g)\,w_i(g),
\]
where $h_0\in H_0$ and $w_\ell(g)$ means one word over the alphabet $\{\beta_\ell,\gamma_\ell\}$. This representation is not unique, but (\ref{eq:conj}) shows that if~$g$ also admits the representation 
\[
    g=h'_0\,w'_1(g)\cdots w'_{i-1}(g)\,w'_i(g),
\]
we have $w_i(g){w'_i}^{-1}(g)\in \gen{\beta_i,\gamma_i}\cap H_{i-1}\le \gen{H_0,\beta_i,\gamma_i}\cap H_{i-1}  = H_{0}$. 

Define the map
$\psi:H_i \longrightarrow \gen{H_0,\beta_i,\gamma_i}/H_0$ as 
\[
    \psi\big(h_0\,w_1(g)\cdots w_{i-1}(g)\,w_i(g)\big)=w_i(g)H_0\,.
\]
Clearly, $\psi$ is a surjective homomorphism and $H_{i-1} \leq \ker(\psi)$. Since $\ker(\psi) \neq H_i$, and  $H_{i}/H_{i-1}\cong T_i$ is a simple group, $H_{i-1} = \ker(\psi)$.   By the first homomorphism theorem of groups, we get $H_{i}/H_{i-1} \cong \gen{H_0,\beta_i,\gamma_i}/H_0$, i.e., Condition $(\star)$ holds. Combined with the fact that each $H_i/H_{i-1}$ is simple, this implies
Condition (C3) of Definition \ref{def:nice}. We conclude that $\pker{G}$ has a nice decomposition satisfying Condition ($\star$).
\end{proof}

\section{Testing Isomorphism to a Ree Group of Rank One}\label{sec:Ree}
In this section, we study the following problem.

\begin{center}
	\ovalbox{
	\begin{minipage}{12.5 cm} \vspace{2mm}
	
	\noindent\hspace{3mm}{\tt ReeIso}$(\bbgroup)$\hspace{5mm}// $\bbgroup$ {\tt is a black-box group}\\\vspace{-3mm}
	
	\noindent\hspace{3mm} Input: $\bul$ a set of generators of a solvable subgroup $L\le \bbgroup$

	\noindent\hspace{15mm}
	$\bul$ two elements $\beta,\gamma\in \bbgroup$ such that $L$ is normal in $\gen{\beta,\gamma,L}$

        \noindent\hspace{15mm}
	$\bul$ an integer $q$ of the form $3^{2a+1}$ for some $a>0$
    
	\vspace{2mm}
	\noindent\hspace{3mm} Output: yes if $\gen{\beta,\gamma,L}/L$ is isomorphic to $\Ree$; no otherwise
	\vspace{2mm}
	\end{minipage}
	}\vspace{2mm}
	\end{center}

Here is the main result of this section.

\begin{theorem}\label{prop:Ree}
    The problem {\tt ReeIso}$(\bbgroup)$ is in $\QCMA$.
\end{theorem}
We prove \cref{prop:Ree} in \cref{sub:proofmain}, after introducing in \cref{sub:Ree2} tools to handle $\Ree$ in its standard representation.

\subsection{Constructive membership for the Ree group in its standard representation}\label{sub:Ree2}
Consider the standard representation $\Ree=\gen{\Gamma_1,\Gamma_2,\Gamma_3}$ introduced in \cref{sub:Ree}.
We will need to implement constructive membership in this representation, i.e., given a matrix $M\in\GL(7,q)$ that belongs to $\gen{\Gamma_1,\Gamma_2,\Gamma_3}$, find a straight-line program over $\{\Gamma_1,\Gamma_2,\Gamma_3\}$ that reaches $M$. Since $\field_q$ has characteristic 3, we can use the following randomized algorithm for constructive membership in matrix groups of odd characteristic from \cite{Babai+STOC09}:
\begin{theorem}[Theorem 2.3 in \cite{Babai+STOC09}]\label{th1}
    There exists a randomized polynomial-time algorithm that solves the constructive membership problem in matrix groups of odd characteristic, given number-theoretic oracles.
\end{theorem}
The output of the randomized algorithm of \cref{th1} is not unique: the straight-line program output by the algorithm depends on the random bits used by the algorithm. Since in our applications we will need to specify a unique output, we consider the algorithm of \cref{th1} as a deterministic algorithm that receives as an auxiliary input a seed of random bits (which we denote by $\seed$). The number-theoretic oracles referred to in \cref{th1} are oracles for integer factoring and discrete logarithm. Since these two tasks can be implemented in polynomial time using a quantum computer \cite{Shor97}, this gives a polynomial-time quantum algorithm for constructive membership in $\Ree$, which we denote by $\Qq_q(\seed)$. Here is the precise statement.
\begin{corollary}\label{cor}
    For any $q$, 
    there exists a collection of polynomial-time quantum algorithms 
    \[\Big\{\Qq_q(\seed)\:|\:\seed\in\{0,1\}^{\poly(\log q)}\Big\}\] that receive as input a matrix $M\in\GL(7,q)$ and satisfy the following condition for any $M\in\Ree$:
    For a fraction at least $1-10^{-7}$ of the $\seed$'s, $\Qq_q(\seed)$ outputs ``success'' with probability at least $1-1/\poly(q)$. When $\Qq_q(\seed)$ outputs ``success'' it also outputs a straight-line program over $\{\Gamma_1,\Gamma_2,\Gamma_3\}$ reaching~$M$. This straight-line program depends only on $M$ and $\seed$ (i.e., it does not depend on the measurement outcomes of the quantum algorithm).
\end{corollary}
Let $\valid$ denote the subset of $\Ree$ containing all the $M\in\Ree$ such that $\Qq_q(\seed)$ on input $M$ outputs ``success'' with probability at least $1-1/\poly(q)$, for the same polynomial as in \cref{cor}. Let $\good$ denote the $\seed$'s such that 
\[
    \Pr_{M\in\Ree}[M\in\valid]\ge 1-10^{-5}\,.
\]
The following claim follows from \cref{cor} by a counting argument.
\begin{claim}\label{claim1}
    When $\seed$ is taken uniformly at random, 
    \[
    \Pr[\seed\in\good]\ge 0.99\,.
    \]
\end{claim}

\begin{proof}
    Let $\delta=2^{\poly(\log q)}$ denote the total number of seeds $\seed$. The number of pairs $(\seed,M)$ such that $M\in\valid$ is $(1-10^{-7})\delta|\Ree|$. Note that for each $\seed\notin \good$ there are at most $(1-10^{-5})\abs{\Ree}$ pairs $(\seed,M)$ with $M\in\valid$, while for each $\seed\in\good$ there are (obviously) at most $\abs{\Ree}$ pairs $(\seed,M)$ with $M\in\valid$. We thus have
    \[
        \abs{\good}\abs{\Ree}+(\delta-\abs{\good})(1-10^{-5})\abs{\Ree}\ge (1-10^{-7})\delta|\Ree|,
    \]
    which implies
    \[
        \frac{\abs{\good}}{\delta}\ge \frac{10^{-5}-10^{-7}}{10{^{-5}}}=0.99,
    \]
    as claimed.
\end{proof}

\subsection{Proof of \cref{prop:Ree}}\label{sub:proofmain}
We are now ready to give the proof of the main result of this section.

\begin{proof}[Proof of \cref{prop:Ree}]
We write $K=\gen{\beta,\gamma,L}$.
We use the standard representation $\Ree=\gen{\Gamma_1,\Gamma_2,\Gamma_3}$ introduced in \cref{sub:Ree} and studied in \cref{sub:Ree2}. 

\paragraph{Merlin's witness.}
Merlin sends three elements $g_1,g_2,g_3\in K$ along with certificates of the following memberships: $g_i\in K$ for all $i\in\{1,2,3\}$, $\beta\in \gen{g_1,g_2,g_3,L}$ and $\gamma\in \gen{g_1,g_2,g_3,L}$.

\paragraph{Definition of the isomorphism candidate.}
Before explaining Arthur's checking procedure, we introduce a map $f_\seed\colon \Ree\to K/L$ defined by the three elements $g_1,g_2,g_3$ sent by Merlin and a binary string $\seed$ that will be later chosen by Arthur. 

Given as input $x\in \valid$, Algorithm $\Qq_q(\seed)$ described in \cref{sub:Ree2} outputs with probability at least $1-1/\poly(q)$ a straight-line program $(w_1,w_2,\ldots,w_s)$ over $\{\Gamma_1,\Gamma_2,\Gamma_3\}$ reaching $x$, where $s=\polylog(|\Ree|)$.
Remember that this means that $w_s=x$ and each $w_i$ is either 
\begin{itemize}
    \item[(i)] a member of $\{\Gamma_1,\Gamma_2,\Gamma_3\}$, or 
    \item[(ii)] an element of the form $w_j^{-1}$ or $w_jw_k$ from some $j,k<i$.
\end{itemize}
From the output $(w_1,w_2,\ldots,w_s)$ of $\Qq_q(\seed)$ on input $x$, we define a new straight-line program $(w'_1,w'_2,\ldots,w'_s)$ over $\{g_1,g_2,g_3\}$ by replacing $\Gamma_1$ by $g_1$, $\Gamma_2$ by $g_2$ and $\Gamma_3$ by $g_3$ in each Case~(i). 
For instance, the straight-line program 
\[
    (w_1=\Gamma_1,w_2=\Gamma_2,w_3=w_1w_1,w_4=w_3w_2, w_5=w_4^{-1},w_6=\Gamma_3,w_7=w_5w_6)
\] 
reaching the element $(\Gamma_1\Gamma_1\Gamma_2)^{-1}\Gamma_3$ becomes the straight-line program 
\[
    (w'_1=g_1,w'_2=g_2,w'_3=w'_1w'_1,w'_4=w'_3w'_2, w'_5={w'}_{\!\!4}^{-1},w'_6=g_3,w'_7=w'_5w'_6)
\]  
reaching the element $(g_1g_1g_2)^{-1}g_3$. We denote by $g_\seed(x)$ the element of $K$ reached by $(w'_1,w'_2,\ldots,w'_s)$. 

Define the map $f_\seed\colon \Ree\to K/L$ as follows: for any $x\in\Ree$,
\[
    f_\seed(x)=
    \left\{
        \begin{array}{cl}
             g_\seed(x) L &\textrm{ if } x\in\valid\,,\\
             L&\textrm{ otherwise}\,.
         \end{array}
    \right.
\]
We state the following elementary, but crucial, property of this map.

\begin{claim}\label{claim2}
    For any $g_1,g_2,g_3\in K$, if there exists a homomorphism $\varphi\colon\Ree\to K/L$ such that $\varphi(\Gamma_i)=g_iL$ for each $i\in\{1,2,3\}$, then $f_\seed(x)=\varphi(x)$ for any $x\in\valid$.
\end{claim}

\paragraph{Arthur's checking procedure.}
Arthur's main verification procedure is the procedure $\htest$ described below, which uses $\majproc(x)$ as a subprocedure. This verification procedure uses (at \cref{step:sample1} of $\htest$ and \cref{step:sample2} of $\majproc(x)$) the sampling algorithm from \cref{th_babai} (with $\varepsilon=1/\poly(\abs{\bbgroup})$) to sample a random element a constant number of times. As explained below, it also uses a constant number of times  Watrous' algorithm (\cref{prop:Watrous}) for membership in a solvable group. 
Since the sampling algorithm only performs approximate sampling and the second algorithm only succeeds with high probability, this may introduce some errors. These errors are nevertheless exponentially small, and thus have a negligible impact on the overall success probability. Additionally, when applying Algorithm $\Qq_q(\seed)$ on an input $M\in\valid$, the error probability (which is exponential small) also has a negligible impact on the overall success probability. For simplicity, we will thus simply ignore all these failure probabilities in the discussion below.

\SetKwFor{RepTimes}{repeat}{:}{end for}
\begin{center}
    \begin{minipage}{13 cm} \vspace{2mm}
    \begin{algorithm}[H]
        \SetAlgorithmName{}{}{List of Algorithms}
        \nonl \hspace{-4mm}$\htest$ 
        \hspace{4mm}\tcp*[l]{Checks if $K/L$ is isomorphic to $\Ree$}

        Choose $\seed$ uniformly at random\,;

        \RepTimes{$\mathrm{12\, times}$}
        {
            Take two elements $r_1$ and $r_2$ uniformly at random from $\Ree$; \label{step:sample1}

            \lIf{$g_\seed(r_1r_2)g_\seed(r_2)^{-1}g_\seed(r_1)^{-1}\notin L$}{output ``no''\,;}\label{step:htest}
             
        }

        \lIf{$K = L$ \,{\bf or}\, $\mathrm{Merlin's\,membership\, certificates\,are\,incorrect}$}{output ``no''\,;}\label{step:ttest}

        \lIf{$\majproc(\Gamma_i)= g_iL\,\, \mathrm{for\,all}\,\, i\in\{1,2,3\}$}{output ``yes''\,;}\label{step:btest}
        \lElse{output ``no''\,;}
        \end{algorithm}
    \end{minipage}
\end{center}

\begin{center}
    \begin{minipage}{14.8 cm} \vspace{2mm}
    \begin{algorithm}[H]
        \SetAlgorithmName{}{}{List of Algorithms}
        \nonl \hspace{-4mm}$\majproc(x)$ 
        \hspace{4mm}\tcp*[l]{computes $\phi(x)$ for the $\phi$ of \cref{lemma:homo}}

        $s\gets 50$\,;
        
        \For{$i$ {\bf from }$1$ \KwTo $s$}
        {
            Take an element $r$ uniformly at random from $\Ree$\,; \label{step:sample2}
        
            $h_i=g_\seed(xr)g_\seed(r^{-1})$\,;\label{step:maj1}
        }
        {\bf return} the coset of $L$ that appears the most frequently among $h_1L,h_2L,\ldots, h_sL$  (breaking ties arbitrarily)\,;\label{step:maj2}

        \end{algorithm}
    \end{minipage}
\end{center}

At \cref{step:ttest} of $\htest$, to check if $K=L$ we only need to check if $\beta\in L$ and $\gamma\in L$, which can be done using \cref{prop:Watrous}. Checking if Merlin's certificates are correct is straightforward. 
At \cref{step:htest}, we compute 
$g_\lambda(r_1)$, $g_\lambda(r_2)$ and 
$g_\seed(r_1r_2)$ in polynomial time by decomposing $r_1$, $r_2$ and $r_1r_2$ using Algorithm $\Qq_q(\seed)$, compute $g_\seed(r_1r_2)g_\seed(r_2)^{-1}g_\seed(r_1)^{-1}$ using the black box for the group $\bbgroup$, and test membership in $L$. At \cref{step:btest}, we use \cref{prop:Watrous} to check if $\majproc(\Gamma_i)=g_i L$. 

At \cref{step:maj1} of $\majproc$, we compute $g_\seed(xr)$ and $g_\seed(r^{-1})$ in polynomial time by using Algorithm $\Qq_q(\seed)$, and compute $g_\seed(xr)g_\seed(r^{-1})$ using the black box for the group $\bbgroup$. 
At \cref{step:maj2} of $\majproc$, we use \cref{prop:Watrous} to compare the cosets and select the one that appears the most frequently.


\paragraph{Completeness.}
If $K/L\cong\Ree$, then there exists an isomorphism $\varphi\colon \Ree\to K/L$. We assume below that $\seed\in\good$, which happens with probability at least 0.99 (\cref{claim1}).

Merlin sends $g_1,g_2, g_3\in K$ such that $g_iL=\varphi(\Gamma_i)$ for each $i\in\{1,2,3\}$, as well as a correct certificate of the membership $g_i\in K$, for each $i\in\{1,2,3\}$. Since $\varphi$ is surjective, Merlin can also send correct certificates of the memberships $\beta\in \gen{g_1,g_2,g_3,L}$ and $\gamma\in \gen{g_1,g_2,g_3,L}$. 

Since $\seed\in\good$, we know that $r_1,r_2$ and $r_1r_2$ are all in $\valid$ with probability at least 0.99997, in which case 
$f_\seed(r_1r_2)=g_\seed(r_1r_2)L$, $f_\seed(r_1)=g_\seed(r_1)L$ and $f_\seed(r_2)=g_\seed(r_2)L$ hold.
\cref{claim2} also implies that $f_\seed(r_1r_2)=\varphi(r_1r_2)$, $f_\seed(r_1)=\varphi(r_1)$ and $f_\seed(r_2)=\varphi(r_2)$.
Since $\varphi(r_1r_2)=\varphi(r_1)\varphi(r_2)$, we obtain $g_\seed(r_1r_2)g_\seed(r_2)^{-1}g_\seed(r_1)^{-1}\in L$.
We conclude that Procedure $\htest$ does not output ``no'' at \cref{step:htest} with probability at least $1-12\cdot 0.00003>0.99$. 

From a similar argument, at \cref{step:maj1} of $\majproc(x)$ we have $h_iL=f_\seed(xr) f_\seed(r^{-1})=\varphi(xr)\varphi(r^{-1})=\varphi(x)$ with probability at least $0.99998$.
Among the 50 trials, we always get $h_iL=\varphi(x)$ with probability at least $0.999$,
in which case we have $\majproc(x)=\varphi(x)$. In particular,
we have $\majproc(\Gamma_i)=\varphi(\Gamma_i)=g_iL$ for all $i\in\{1,2,3\}$ with probability at least $(0.999)^3>0.99$.  

The overall probability that $\htest$ outputs ``yes'' is thus at least $(0.99)^3>2/3$. 

\paragraph{Soundness.}
Now consider the case $K/L\not\cong \Ree$. In the following, we assume that $\seed\in\good$, which happens with probability at least 0.99 (\cref{claim1}). Assume for now that $\valid=\Ree$.

Consider first the case
\[
     \Pr_{r_1,r_2\in \Ree}\:[f_\seed(r_1r_2) = f_\seed(r_1)f_\seed(r_2)]< 9/10.
\]
In this case, $\htest$ outputs ``no'' at \cref{step:htest}  at least once during the 12 iterations with probability at least $1-(9/10)^{12}>0.7$. 

Now consider the case 
\[
     \Pr_{r_1,r_2\in \Ree}\:[f_\seed(r_1r_2) = f_\seed(r_1)f_\seed(r_2)]\ge 9/10.
\]
If $K=L$ or Merlin's certificates are incorrect,  $\htest$ outputs ``no'' at \cref{step:ttest}. 
We thus assume below that $K\neq L$ (i.e., $K/L\neq\tgroup$) and Merlin's certificate are correct (i.e., $\gen{g_1,g_2,g_3,L}=K$). 
\cref{lemma:homo} shows that there exists a homomorphism $\phi\colon\Ree\to K/L$ such that for any $x\in\Ree$,
\begin{align*}
\Pr_{r\in \Ree}\:[f_\seed(xr)f_\seed(r^{-1})=\phi(x)] 
&=\Pr_{r\in \Ree}\:[f_\seed(xr)f_\seed(r^{-1})=\phi(xr)\phi(r^{-1})]\\
&=1-\Pr_{r\in \Ree}\:[f_\seed(xr)f_\seed(r^{-1})\neq\phi(xr)\phi(r^{-1})]\\
&\geq 1-\Pr_{r\in \Ree}\:[f_\seed(xr)\neq \phi(xr)] 
- \Pr_{r\in \Ree}\:[f_\seed(r^{-1})\neq \phi(r^{-1})]\\
&\geq 1-\frac{1}{10}-\frac{1}{10}\\
&= \frac{8}{10}.
\end{align*}
Among 50 trials, the expected number of times we get $\phi(x)$ at \cref{step:maj2} of $\majproc(x)$ is thus at least $(8/10)\cdot 50$. 
From Chernoff's bound, this implies that $\phi(x)$ appears at least 26 times among the 50 times with probability at least 
\[
    1-\exp\left(-\frac{(7/20)^2 \cdot (8/10)\cdot 50}{2}\right)>0.9,
\] 
in which case we have $\majproc(x)=\phi(x)$.
 In particular, 
\begin{equation}\label{eq:cond}
    \majproc(\Gamma_i)=\phi(\Gamma_i) \textrm{ for all }i\in\{1,2,3\}
\end{equation} 
holds with probability at least 0.7. If \cref{eq:cond} holds, then there should be an index $i\in\{1,2,3\}$ such that $\majproc(\Gamma_i)\neq g_iL$. Otherwise, $\phi$ would be a surjective homomorphism from $\Ree$ to $K/L$, and thus an isomorphism since $\Ree$ is a simple group and $K/L\neq\tgroup$, which contradicts the assumption $K/L\not\cong \Ree$. The probability that $\htest$ outputs ``no'' at \cref{step:btest} is thus at least 0.7.

We actually have $\valid\neq \Ree$. The probability that the arguments of the $12\cdot 3+50\cdot 2=136$ calls to the function $g_\lambda$ performed by $\htest$ are all in the set $\valid$ is nevertheless at least $1-132\cdot 10^{-5}> 0.98$. The overall probability that $\htest$ outputs ``no'' is thus at least $0.7-0.01-0.02>2/3$. 
\end{proof}

\section{Proof of \cref{th:order}}\label{sec:main}
In this section, we prove \cref{th:order}, i.e., we show that checking that $\abs{G}=m$ is in $\QCMA$, where $G=\gen{g_1,\ldots,g_k}\le\bbgroup$ and $m$ are the inputs of Group Order Verification (here $\bbgroup$ is a black-box group). We divide the proof into two parts: checking that $m$ divides~$|G|$ (\cref{prop:lb}, which is the easy part) and checking that $|G|$ divides $m$ (\cref{prop:ub}, which is the hard part). \cref{prop:lb} and \cref{prop:ub} together immediately imply \cref{th:order}.

\subsection{Checking that \boldmath{$m$} divides the order}
Adapting the classical strategy from \cite[Section~9]{Babai+FOCS84} to the quantum setting by replacing oracles for ``independence testing'' by efficient quantum algorithms dealing with solvable groups (\cref{prop:Watrous}), we obtain the following result.
\begin{proposition}\label{prop:lb}
There exists a $\QCMA$ protocol that checks if $m$ divides $|G|$.
\end{proposition}
\begin{proof}
Merlin sends to Arthur the factorization of $m$ as a product of primes $m=p_1^{t_1}\cdots p_r^{t_r}$ and for each $i\in[r]$ a set of generators of a subgroup $G_i\le G$, together with membership certificates certifying that $G_i$ is a subgroup of $G$. Arthur accepts if and only if the following three conditions hold:
\begin{itemize}
    \item[(i)] the factorization of $m$ is correct;
    \item[(ii)] the membership certificates are correct;
    \item[(iii)] 
    $\abs{G_i}=p_i^{t_i}$ for each $i\in[r]$.
\end{itemize}
Conditions (i) and (ii) can be checked in deterministic polynomial time. Since any group of order $p^t$ for some prime $p$ and some positive integer $t$ is solvable, Arthur can check Condition (iii) in quantum polynomial time by first checking solvability using \cref{lemma:sol} and then computing the order using \cref{prop:Watrous}.

The completeness and soundness of this protocol follow from \cref{lemma:nilpotent}.
\end{proof}

\subsection{Checking that the order divides \boldmath{$m$}}
The hard part is to show that the order divides $m$. Combining all the techniques developed in this paper, we show the following result.

\begin{proposition}\label{prop:ub}
There exists a $\QCMA$ protocol that checks if $|G|$ divides $m$.
\end{proposition}
\begin{proof}
We first describe the $\QCMA$ protocol and then analyze its correctness and soundness. Note that the completeness and soundness of the $\QCMA$ protocol of \cref{prop:Ree} can be amplified using standard techniques so that the completeness becomes $1-1/\poly(\abs{\bbgroup})$ and the soundness becomes $1/\poly(\abs{\bbgroup})$. In the following, we implicitly assume that the completeness and soundness have been amplified. 

\paragraph{Merlin's witness.}
Merlin sends to Arthur a positive integer $m_1$ and elements $h_1,\ldots, h_n,k_1,\ldots,k_r$, $\beta_1,\ldots,\beta_s$, $\gamma_1,\ldots,\gamma_s\in \bbgroup$, for some $n,r,s=O(\log\abs{\bbgroup})$. 
We write below 
\begin{align*}
    H_0&=\gen{h_1,\ldots,h_n}\,,\\
    K&=\gen{k_1,\ldots,k_r}\,,\\
    H_i&=\gen{H_0,\beta_1,\ldots,\beta_i,\gamma_1,\ldots,\gamma_i} \:\textrm { for each } i\in[s]\,.
\end{align*}
Merlin also sends to Arthur the following information:
\begin{enumerate}[(i)]
    \item membership certificates certifying that $h_1,\ldots, h_n, \beta_1,\ldots,\beta_s$, $\gamma_1,\ldots,\gamma_s$ are elements in $K$\,;
    \item normality certificates certifying that $K\ns G$, $H_s\ns K$, and $H_{i-1}\ns H_i$ for each $i\in[s]$\,; 
    \item a solvability certificate certifying that $K/H_s$ is a solvable group of order divising $m_1$\,;
    \item a multiset $\mathcal{S}$ of simple groups with polylogarithmic-size presentation (each simple group is given by its standard name) and 
    the certificate from \cref{th:p3} certifying that $\mathcal{S}$ is the set of composition factors of $G/K$\,;
    \item for each $i\in[s]$, the standard name $z_i $ of a simple group\,;
    \item for each $i\in[s]$ such that $z_i$ corresponds to the standard name of a finite simple group other than a Ree group of rank one, 
    the certificate from \cref{th:p2} certifying that $\gen{H_0,\beta_i,\gamma_i}/H_0$ is isomorphic to $\group{z_i}$\,;
    \item for each $i\in[s]$ such that $z_i$ corresponds to the standard name of some Ree group of rank one $\Ree$, the certificate from \cref{prop:Ree} certifying that $\gen{H_0,\beta_i,\gamma_i}/H_0$ is isomorphic to $\Ree$\,.
\end{enumerate}

\paragraph{Terminology.} We use the following terminology below: we say that a certificate in (iv) or (vi) is correct if the checking procedure of \cref{th:p3} or \cref{th:p2}, respectively, outputs ``yes'' on this certificate.  We say that the certificate in (vii) is correct if the checking procedure of \cref{prop:Ree} outputs ``yes'' on this certificate with probability at least $1-1/\poly(\abs{\bbgroup})$. Note that for \cref{th:p2}, \cref{th:p3} and \cref{prop:Ree}, the existence of a correct certificate guarantees that the input is a yes-instance.

\paragraph{Arthur's checking procedure.}
Arthur first computes a set of generators of $\pker{G}$, the multiset of composition factors of $G/\pker{G}$ and the order $\abs{G/\pker{G}}$, which can be done in randomized polynomial time as discussed in \cref{sub:bb}. We write $m_{2}=\abs{G/\pker{G}}$.
Arthur then checks that
\begin{enumerate}[(1)]
    \item $H_0$ is a solvable normal subgroup of $\pker{G}$\,;
    \item $k_1,\ldots,k_r\in\pker{G}$\,;
    \item $\mathcal{S}$ is the multiset of composition factors of $G/\pker{G}$\,;
    \item each membership certificate in (i) is correct\,;
    \item each normality certificate in (ii) is correct, i.e., it certifies all the inclusions of \cref{eq:norm}\,;
    \item the solvability certificate in (iii) is correct, i.e., it certifies all the conditions of \cref{def:sol}\,;
    \item the certificate in (iv) is a correct certificate for \cref{th:p3};
    \item the certificates in
    (vi) are correct certificates for \cref{th:p2};
    \item the certificates in
    (vii) are correct certificates for \cref{prop:Ree};
    \item
    the product $|H_0|\cdot \abs{\group{z_1}}\cdots \abs{\group{z_s}}\cdot m_1\cdot m_2$
    divides $m$\,.
\end{enumerate}
Item (1) can be implemented using \cref{lemma:sol} and \cref{cor:Watrous}. Item (2) can be implemented in deterministic polynomial time using the efficient procedure for membership in $\pker{G}$ of \cref{sub:bb}. Item~(3) is trivial to check. Item (4)$\sim$(8) can be checked in deterministic polynomial time from the discussion in \cref{sub:black} (for (7) and (8) we use \cref{th:p3} and \cref{th:p2}, respectively). Item (9) can be checked (with high probability) in quantum polynomial time, from \cref{prop:Ree}. To check Item (10), we just need to compute $|H_0|$, which can be done with high probability using \cref{prop:Watrous}.

\paragraph{Remark 1.}
We cannot ask directly Merlin to prove that $K=\pker{G}$ since Merlin does not know the elements of the generating set of $\pker{G}$ computed by Arthur (generators of $\pker{G}$ can only be computed in randomized polynomial time, not in deterministic polynomial time). Instead, we ask Merlin to send $\mathcal{S}$ and check that $K=\pker{G}$ using Tests (2), (3) and (7).

\paragraph{Remark 2.}
Even when $K/H_s$ is solvable, we cannot use Watrous' quantum algorithm (\cref{prop:Watrous}) to compute its order since the group $K/H_s$ does not have unique encoding. This is why we ask Merlin to certify that the order of $K/H_s$ divides $m_1$ using the solvability certificate (iii), which can be checked classically even without unique encoding.

\paragraph{Completeness.}
Assume that $|G|$ divides $m$. From \cref{th:group}, there exist a normal subgroup $H_0\ns \pker{G}$ and $2s$ elements $\beta_1,\ldots,\beta_s,\gamma_1,\ldots,\gamma_s\in \pker{G}$ such that, when defining 
    \begin{align*}
        H_i&=\gen{H_0,\beta_1,\ldots,\beta_i,\gamma_1,\ldots,\gamma_i}
    \end{align*}
for each $i\in[s]$, Conditions (C1), (C2), (C3) of \cref{def:nice} and Condition ($\star$) of \cref{th:group} are satisfied. We have 
\[
\abs{G}=\abs{H_0}\cdot \abs{H_1/H_0}\cdots\abs{H_s/H_{s-1}}\cdot\abs{\pker{G}/H_s}\cdot\abs{G/\pker{G}}.
\]
Merlin sends generators of this subgroup $H_0$, generators of $K=\pker{G}$, these $2s$ elements $\beta_1,\ldots,\beta_s,\gamma_1,\ldots,\gamma_s$, and the integer $m_1=\abs{\pker{G}/H_s}$. Merlin sends the multiset $\mathcal{S}$ of composition factors of $G/\pker{G}$. Each $z_i$ sent by Merlin is the standard name of the simple group $\gen{H_0,\beta_i,\gamma_i}/H_0$.

The existence of correct normality certificates for (ii) follows from the normality of $H_i$, $H_s$ and $\pker{G}$. The existence of a correct solvability certificate for (iii) follows from Condition (C2) of \cref{def:nice}. The existence of a correct certificate for (iv) follows from \cref{th:p3} combined with \cref{th:smallk} and Property (d) of the Babai-Beals filtration described in \cref{sub:bb}, which guarantee that the composition factors of $G/\pker{G}$ (i.e., the simple groups in $\mathcal{S}$) have a short presentation. The existence of correct certificates of (vi) and (vii) follow from \cref{th:p2} and \cref{prop:Ree}, respectively. 

With the above choices, checking Items (1) and (9) succeeds with high probability, while checking Items (2)$\sim$(8) always succeed.
From Condition ($\star$), we know that $H_i/H_{i-1}$ is isomorphic to $\gen{H_0,\beta_i,\gamma_i}/H_0$ for all $i\in[s]$.
We thus have
\[
    \abs{G}=\abs{H_0}\cdot \abs{\group{z_1}}\cdots \abs{\group{z_s}}\cdot m_1\cdot m_2.
\]
Since $|G|$ divides $m$, this quantity divides $m$, and thus checking Item (10) also succeeds with high probability.

\paragraph{Soundness.}
Assume that $|G|$ does not divide $m$. If Item (1) or (9) is not true, Arthur rejects with high probability. If Items (2)$\sim$(7) are not all true, Arthur always rejects. In the following, we thus assume that Items (1)$\sim$(9) are all true.

Item (1) guarantees that $H_0$ is a solvable normal subgroup of $\pker{G}$.
Item (2) guarantees that~$K$ is a subgroup of $\pker{G}$.
Items (3) and (7) guarantee that the multiset of composition factors of $G/K$ matches the multiset of composition factors of $G/\pker{G}$, which implies $\abs{K}=\abs{\pker{G}}$. We thus have $K=\pker{G}$.

Item (4) and (5) further guarantee that
\[
    \tgroup\ns H_0\ns H_1\ns \cdots\ns H_s\ns \pker{G}\ns G.
\]
Item (6) guarantees that $\abs{\pker{G}/H_s}$ divides $m_1$. Item (8) and (9) guarantee that $\gen{H_0,\beta_i,\gamma_i}/H_0\cong\group{z_i}$ for each $i\in[s]$.  \cref{th:gmain} then implies that $\abs{H_i/H_{i-1}}$ also divides $\abs{\group{z_i}}$, for each $i\in[s]$. 
We conclude that $\abs{G}$ must divide the quantity 
\[
    \abs{H_0}\cdot \abs{\group{z_1}}\cdots \abs{\group{z_s}}\cdot m_1\cdot m_2.
\]
Since $|G|$ does not divide $m$, this implies that Item (10) must fail whenever the computation of $\abs{H_0}$ is correct, which happens with high probability.
\end{proof}

\section{Proofs of the Other Results}\label{sec:other}
In this section, we discuss how to derive the other results of \cref{table:results}. 

We first give the formal definition of the eight problems introduced in Section \ref{subsec:results}. 

\begin{center}
\underline{Group Isomorphism}\\[2mm]
\begin{flushleft}
\begin{tabular}{ll}
Instance: & Elements $g_1,\,\ldots,\,g_k$ in some group $\bbgroup$, elements $h_1,\,\ldots,\,h_\ell$ in some group $\bbgroupb$.\\
Question: & Are $\gen{g_1,\,\ldots,\,g_k}$ and $\gen{h_1,\,\ldots,\,h_\ell}$ isomorphic?
\end{tabular}
\end{flushleft}
\end{center}

\begin{center}
\underline{Homomorphism}\\[2mm]
\begin{flushleft}
\begin{tabular}{ll}
Instance: & Elements $g_1,\,\ldots,\,g_k$ in some group $\bbgroup$, elements $h_1,\,\ldots,\,h_k$ in some group $\bbgroupb$.\\
Question: & Is there a homomorphism $\phi:\gen{g_1,\,\ldots,\,g_k}\to\gen{h_1,\,\ldots,\,h_k}$ such that $\phi(g_i)=h_i$\\ 
&\hspace{106mm}  for each $i\in[k]$?
\end{tabular}
\end{flushleft}
\end{center}


\begin{center}
\underline{Minimal Normal Subgroup}\\[2mm]
\begin{flushleft}
\begin{tabular}{ll}
Instance: & Elements $g_1,\,\ldots,\,g_k$ and $h_1,\,\ldots,\,h_\ell$ in some
group $\bbgroup$.\\
Question: & Is $\langle h_1,\ldots,h_\ell\rangle$ a minimal normal subgroup of
$\langle g_1,\ldots,g_k\rangle$?
\end{tabular}
\end{flushleft}
\end{center}

\begin{center}
\underline{Proper Subgroup}\\[2mm]
\begin{flushleft}
\begin{tabular}{ll}
Instance: & Elements $g_1,\,\ldots,\,g_k$ and $h_1,\,\ldots,\,h_\ell$ in
some group $\bbgroup$.\\
Question: & Is $\,\langle h_1,\:\ldots,\:h_\ell\rangle\,$ a proper subgroup of
$\langle g_1,\ldots,g_k\rangle$?
\end{tabular}
\end{flushleft}
\end{center}

\begin{center}
\underline{Simple Group}\\[2mm]
\begin{flushleft}
\begin{tabular}{ll}
Instance: & Elements $g_1,\ldots,g_k$ in some group $\bbgroup$.\\
Question: & Is $\langle g_1,\ldots,g_k\rangle$ a simple group?
\end{tabular}
\end{flushleft}
\end{center}

\begin{center}
\underline{Intersection}\\[2mm]
\begin{flushleft}
\begin{tabular}{ll}
Instance: & Elements $g_1,\:\ldots,\:g_k$, $h_1,\:\ldots,\:h_\ell$, and
$a_1,\ldots,a_t$ in some group $\bbgroup$.\\
Question: & Is $\langle a_1,\ldots,a_t\rangle$ equal to the intersection
of $\langle g_1,\ldots,g_k\rangle$ and $\langle h_1,\ldots,h_\ell\rangle$?
\end{tabular}
\end{flushleft}
\end{center}

\begin{center}
\underline{Centralizer}\\[2mm]
\begin{flushleft}
\begin{tabular}{ll}
Instance: & Elements $g_1,\ldots,g_k$, $h_1,\ldots,h_\ell$ and $a$ in some
group $\bbgroup$.\\
Question: & Is $\langle h_1,\ldots,h_\ell\rangle$ equal to the centralizer of
$a$ in $\langle g_1,\ldots,g_k\rangle$?
\end{tabular}
\end{flushleft}
\end{center}

\begin{center}
\underline{Maximal Normal Subgroup}\\[2mm]
\begin{flushleft}
\begin{tabular}{ll}
Instance: & Elements $g_1,\,\ldots,\,g_k$ and $h_1,\,\ldots,\,h_\ell$ in some
group $\bbgroup$.\\
Question: & Is $\langle h_1,\ldots,h_\ell\rangle$ a maximal normal subgroup of
$\langle g_1,\ldots,g_k\rangle$?
\end{tabular}
\end{flushleft}
\end{center}

We give below the proofs of Corollaries \ref{cor:result1}, \ref{cor:result2} and \ref{cor:result3}. 

\addtocounter{section}{-5}
\addtocounter{corollary}{2}
\begin{corollary}[repeated]
Group Isomorphism is in the complexity class $\QCMA$.   
\end{corollary}
\begin{proof}
Group Isomorphism can be reduced to the Group Order Verification as follows (see, \cite[Proposition 4.9]{Babai92}): Merlin guesses $k$ elements $h'_1,\,\ldots,\,h'_k$ from $\gen{h_1,\,\ldots,\,h_\ell}$ such that $h'_i=\phi(g_i)$ for each $i\in[k]$, for some isomorphism $\phi\colon\gen{g_1,\,\ldots,\,g_k}\to \gen{h_1,\,\ldots,\,h_\ell}$. Merlin also guesses $m=\abs{\gen{g_1,\,\ldots,\,g_k}}$, as well as membership certificates certifying that $\gen{h'_1,\,\ldots,\,h'_k}=\gen{h_1,\,\ldots,\,h_\ell}$. 

Consider the subgroup $K=\gen{(g_1,h'_1), \cdots,(g_k,h'_k)}$ of the group $\gen{g_1,\,\ldots,\,g_k} \times \gen{h_1,\,\ldots,\,h_\ell}$. 
Arthur checks that the membership certificates are correct and checks that $m=\abs{\gen{g_1,\,\ldots,\,g_k}}=\abs{\gen{h_1,\,\ldots,\,h_\ell}}=\abs{K}$ using the $\QCMA$ protocol for Group Order Verification.  
\end{proof}
\begin{corollary}[repeated]
Homomorphism, Minimal Normal Subgroup, Proper Subgroup and Simple Group are in the complexity class $\QCMA \cap\coQCMA $.   
\end{corollary}
\begin{proof}
Homomorphism and Minimal Normal Subgroup can be reduced to Group Order in  
polynomial time (see, \cite[Corollary 12.1]{Babai92} and its proof). 
The claim thus follows from \cref{cor:order}.


To show that Proper Subgroup is in $\QCMA$, 
we follow the strategy of \cite[Section 5]{WatrousFOCS00}. Merlin guesses an element $a \in \gen{g_1,\,\ldots,\,g_k}$ such that $a \not \in \gen{h_1,\,\ldots,\,h_\ell}$ and membership certificates certifying that $a \in \gen{g_1,\,\ldots,\,g_k}$ and $h_i \in \gen{g_1,\,\ldots,\,g_k}$ for each $i\in[\ell]$. Arthur checks that the membership certificates are correct and checks that $a \not \in \gen{h_1,\,\ldots,\,h_\ell}$ using the $\QCMA$ protocol for Group Non-Membership.  

To show that Proper Subgroup is in $\coQCMA$, we observe that $\gen{h_1,\,\ldots,\,h_\ell}$ is not a proper subgroup of $\gen{g_1,\,\ldots,\,g_k}$ if and only if either $\gen{h_1,\,\ldots,\,h_\ell}=\gen{g_1,\,\ldots,\,g_k}$ or there exists an element $h\in \gen{h_1,\,\ldots,\,h_\ell}$ such that $h\notin \gen{g_1,\,\ldots,\,g_k}$. Merlin guesses which case holds. In the first case, he also guesses membership certificates certifying that $\gen{h_1,\,\ldots,\,h_\ell}=\gen{g_1,\,\ldots,\,g_k}$. In the second case, he also guesses an element $h\in \gen{h_1,\,\ldots,\,h_\ell}$ such that $h\notin \gen{g_1,\,\ldots,\,g_k}$. Arthur checks that the membership certificates are correct and checks that $h\notin \gen{g_1,\,\ldots,\,g_k}$ using the $\QCMA$ protocol for Group Non-Membership. 

Watrous \cite[Section 5]{WatrousFOCS00} showed that Simple Group is in $\coQMA$ by using a quantum proof for Group Non-Membership. From \cref{th:GNM}, we can conclude that Simple Group is in $\coQCMA$. \cref{th:p1}, \cref{th:p2} and \cref{prop:Ree} together imply that Simple Group is in $\QCMA$.
\end{proof}

\begin{corollary}[repeated]
Intersection, Centralizer and Maximal Normal Subgroup are in the complexity class $\coQCMA$.   
\end{corollary}
\begin{proof}
Watrous \cite{WatrousFOCS00} showed that Intersection, Centralizer and Maximal Normal Subgroup are in $\coQMA$ by using a quantum proof for Group Non-Membership along with classical proofs for various other properties (see,  \cite[Section 5]{WatrousFOCS00}). By \cref{th:GNM}, Group Non-Membership is in $\QCMA$. This implies that Intersection, Centralizer and Maximal Normal Subgroup are in $\coQCMA$.
\end{proof}

\addtocounter{section}{+5}
\addtocounter{corollary}{-2}
\section*{Acknowledgments}
The authors are grateful to Scott Aaronson, Michael Levet and James Wilson for helpful correspondence, and to Hirotada Kobayashi for many fruitful discussions. The authors are supported by JSPS KAKENHI 22H00522, 24H00071, 24K22293, MEXT Q-LEAP JPMXS0120319794, JST CREST JPMJCR24I4 and JST ASPIRE JPMJAP2302.
\bibliographystyle{plain}
\bibliography{References}

\providecommand{\noopsort}[1]{}
\begin{thebibliography}{10}

\bibitem{Aaronson16}
Scott Aaronson.
\newblock The complexity of quantum states and transformations: From quantum money to black holes.
\newblock ArXiv: 1607.05256, 2016.

\bibitem{Aaronson+ToC07}
Scott Aaronson and Greg Kuperberg.
\newblock Quantum versus classical proofs and advice.
\newblock {\em Theory of Computing}, 3(1):129--157, 2007.
\newblock ArXiv:quant-ph/0604056, 2006.

\bibitem{Aharonov+02}
Dorit Aharonov and Tomer Naveh.
\newblock Quantum {NP} - a survey.
\newblock ArXiv: quant-ph/0210077, 2002.

\bibitem{Baarnhielm14}
Henrik B\"a\"arnhielm.
\newblock Recognising the small {Ree} groups in their natural representations.
\newblock {\em Journal of Algebra}, 416:139--166, 2014.

\bibitem{Babai85}
L{\'{a}}szl{\'{o}} Babai.
\newblock Trading group theory for randomness.
\newblock In {\em Proceedings of the 17th Annual {ACM} Symposium on Theory of Computing (STOC 1985)}, pages 421--429, 1985.

\bibitem{BabaiSTOC91}
L{\'{a}}szl{\'{o}} Babai.
\newblock Local expansion of vertex-transitive graphs and random generation in finite groups.
\newblock In {\em Proceedings of the 23rd Annual {ACM} Symposium on Theory of Computing (STOC 1991)}, pages 164--174, 1991.

\bibitem{Babai92}
L{\'a}szl{\'o} Babai.
\newblock Bounded round interactive proofs in finite groups.
\newblock {\em SIAM Journal on Discrete Mathematics}, 5(1):88--111, 1992.

\bibitem{Babai+99}
L{\'a}szl{\'o} Babai and Robert Beals.
\newblock A polynomial-time theory of black-box groups {I}.
\newblock {\em London Mathematical Society Lecture Note Series}, 260:30--64, 1999.

\bibitem{Babai+STOC09}
L{\'{a}}szl{\'{o}} Babai, Robert Beals, and {\'{A}}kos Seress.
\newblock Polynomial-time theory of matrix groups.
\newblock In {\em Proceedings of the 41st Annual {ACM} Symposium on Theory of Computing (STOC 2009)}, pages 55--64, 2009.

\bibitem{Babai+97}
L{\'{a}}szl{\'{o}} Babai, Albert~J. Goodman, William~M. Kantor, Eugene~M. Luks, and P\'eter~P. P\'alfy.
\newblock Short presentations for finite groups.
\newblock {\em Journal of Algebra}, 194:79--112, 1997.

\bibitem{BaLS}
L{\'{a}}szl{\'{o}} Babai, Eugene~M. Luks, and {\'{A}}kos Seress.
\newblock Computing composition series in primitive groups.
\newblock {\em DIMACS Series in Discrete Mathematics and Theoretical Computer Science}, 11:1--16, 1993.

\bibitem{Babai+FOCS84}
L{\'a}szl{\'o} Babai and Endre Szemer{\'e}di.
\newblock On the complexity of matrix group problems {I}.
\newblock In {\em Proceedings of the 25th Annual Symposium on Foundations of Computer Science (FOCS 1984)}, pages 229--240, 1984.

\bibitem{Ben-David+24}
Shalev Ben-David and Srijita Kundu.
\newblock Oracle separation of {QMA} and {QCMA} with bounded adaptivity.
\newblock In {\em Proceedings of the 51st International Colloquium on Automata, Languages, and Programming (ICALP 2024)}, pages 21:1--21:18, 2024.

\bibitem{BenOr+08}
Michael Ben{-}Or, Don Coppersmith, Michael Luby, and Ronitt Rubinfeld.
\newblock Non-abelian homomorphism testing, and distributions close to their self-convolutions.
\newblock {\em Random Structures \& Algorithms}, 32(1):49--70, 2008.

\bibitem{Cheung+01}
Kevin K.~H Cheung and Michele Mosca.
\newblock Decomposing finite abelian groups.
\newblock {\em Quantum Information and Computation}, 1(3):26--32, 2001.

\bibitem{Dixon+96}
John~D. Dixon and Brian Mortimer.
\newblock {\em Permutation Groups}, volume 163 of {\em Graduate Texts in Mathematics}.
\newblock Springer-Verlag, New York, 1996.

\bibitem{Dummit+04}
David~S. Dummit and Richard~M. Foote.
\newblock {\em Abstract Algebra, Third Edition}.
\newblock John Wiley and Sons, Inc., 2004.

\bibitem{Ettinger+04}
Mark Ettinger, Peter H{\o}yer, and Emanuel Knill.
\newblock The quantum query complexity of the hidden subgroup problem is polynomial.
\newblock {\em Information Processing Letters}, 91(1):43--48, 2004.

\bibitem{Fefferman+18}
Bill Fefferman and Shelby Kimmel.
\newblock {Quantum vs. Classical Proofs and Subset Verification}.
\newblock In {\em Proceedings of the 43rd International Symposium on Mathematical Foundations of Computer Science (MFCS 2018)}, pages 22:1--22:23, 2018.

\bibitem{Grilo+16}
Alex~B. Grilo, Iordanis Kerenidis, and Jamie Sikora.
\newblock {QMA} with subset state witnesses.
\newblock {\em Chicago Journal of Theoretical Computer Science}, pages 4:1--4:17, 2016.

\bibitem{GMPS+2015}
Simon Guest, Joy Morris, Cheryl~E. Praeger, and Pablo Spiga.
\newblock On the maximum orders of elements of finite almost simple groups and primitive permutation groups.
\newblock {\em Transactions of the American Mathematical Society}, 367(11):7665--7694, 2015.

\bibitem{guralnick+00}
Robert~M. Guralnick and William~M. Kantor.
\newblock Probabilistic generation of finite simple groups.
\newblock {\em Journal of Algebra}, 234(2):743--792, 2000.

\bibitem{Derek-Holt-MOF}
Derek Holt.
\newblock Composition factors of primitive components.
\newblock MathOverflow, answer.
\newblock \url{https://mathoverflow.net/q/265254} (version: 2017-03-22).

\bibitem{Hulpke+01}
Alexander Hulpke and {\'A}kos Seress.
\newblock Short presentations for three-dimensional unitary groups.
\newblock {\em Journal of Algebra}, 245(2):719--729, 2001.

\bibitem{Ivanyos+03}
G{\'{a}}bor Ivanyos, Fr{\'{e}}d{\'{e}}ric Magniez, and Miklos Santha.
\newblock Efficient quantum algorithms for some instances of the non-abelian hidden subgroup problem.
\newblock {\em International Journal of Foundations of Computer Science}, 14(5):723--740, 2003.

\bibitem{Jordan+12}
Stephen~P. Jordan, Hirotada Kobayashi, Daniel Nagaj, and Harumichi Nishimura.
\newblock Achieving perfect completeness in classical-witness quantum {Merlin}-{Arthur} proof systems.
\newblock {\em Quantum Information \& Computation}, 12(5-6):461--471, 2012.

\bibitem{Kemper+01}
Gregor Kemper, Frank Lübeck, and Kay Magaard.
\newblock Matrix generators for the {Ree} groups $^2{G}_2(q)$.
\newblock {\em Communications in Algebra}, 29(1):407--413, 2001.

\bibitem{Kitaev99}
Alexei~Yu. Kitaev.
\newblock Quantum {$\NP$}.
\newblock Talk at the Second Workshop on Algorithms in Quantum Information Processing, DePaul University, January 1999.

\bibitem{knill1996}
Emanuel Knill.
\newblock Quantum randomness and nondeterminism.
\newblock ArXiv: quant-ph/9610012, 1996.

\bibitem{LeGall+MFCS18}
Fran\c{c}ois Le~Gall, Tomoyuki Morimae, Harumichi Nishimura, and Yuki Takeuchi.
\newblock Interactive proofs with polynomial-time quantum prover for computing the order of solvable groups.
\newblock In {\em Proceedings of the 43rd International Symposium on Mathematical Foundations of Computer Science (MFCS 2018)}, pages 26:1--26:13, 2018.

\bibitem{LeGallSTACS10}
Fran{\c{c}}ois {Le Gall}.
\newblock An efficient quantum algorithm for some instances of the group isomorphism problem.
\newblock In {\em Proceedings of the 27th International Symposium on Theoretical Aspects of Computer Science ({STACS} 2010)}, pages 549--560, 2010.

\bibitem{Levet}
Michael Levet.
\newblock Personal communication. December 2024.

\bibitem{li+ITCS24}
Xingjian Li, Qipeng Liu, Angelos Pelecanos, and Takashi Yamakawa.
\newblock Classical vs quantum advice and proofs under classically-accessible oracle.
\newblock In {\em Proceedings of the 15th Innovations in Theoretical Computer Science Conference (ITCS 2024)}, pages 72:1--72:19, 2024.

\bibitem{Liu+24}
Jiahui Liu, Saachi Mutreja, and Henry Yuen.
\newblock {QMA} vs. {QCMA} and pseudorandomness.
\newblock In {\em Proceedings of the 57th Annual {ACM} Symposium on Theory of Computing (STOC 2025)}, pages 1520--1531, 2025.

\bibitem{Luks1997}
Eugene~M. Luks.
\newblock Computing the composition factors of a permutation group in polynomial time.
\newblock {\em Combinatorica}, 7(1):87--99, 1987.

\bibitem{Luks1993}
Eugene~M. Luks.
\newblock Permutation groups and polynomial-time computation.
\newblock {\em DIMACS Series in Discrete Mathematics and Theoretical Computer Science}, 11:139--175, 1993.

\bibitem{Natarajan+24}
Anand Natarajan and Chinmay Nirkhe.
\newblock A distribution testing oracle separation between {QMA} and {QCMA}.
\newblock {\em {Quantum}}, 8:1377, 2024.

\bibitem{ODonnel15}
Ryan O'Donnell.
\newblock Lectures notes on quantum computation and information.
\newblock \url{https://www.cs.cmu.edu/~odonnell/quantum15/}, 2015.

\bibitem{Rotman'02}
Joseph~J. Rotman.
\newblock {\em Advanced Modern Algebra}.
\newblock Prentice Hall, Inc., Upper Saddle River, NJ, 2002.

\bibitem{Shor97}
Peter~W. Shor.
\newblock Polynomial-time algorithms for prime factorization and discrete logarithms on a quantum computer.
\newblock {\em SIAM Journal on Computing}, 26(5):1484--1509, 1997.

\bibitem{TCS-068}
Thomas Vidick and John Watrous.
\newblock Quantum proofs.
\newblock {\em Foundations and Trends® in Theoretical Computer Science}, 11(1-2):1--215, 2016.

\bibitem{WatrousFOCS00}
John Watrous.
\newblock Succinct quantum proofs for properties of finite groups.
\newblock In {\em Proceedings of the 41st Annual Symposium on Foundations of Computer Science (FOCS 2000)}, pages 537--546, 2000.

\bibitem{WatrousSTOC01}
John Watrous.
\newblock Quantum algorithms for solvable groups.
\newblock In {\em Proceedings on 33rd Annual {ACM} Symposium on Theory of Computing (STOC 2001)}, pages 60--67, 2001.

\bibitem{Wilson}
James Wilson.
\newblock Personal communication. December 2024.

\bibitem{Wilson09}
Robert Wilson.
\newblock {\em The Finite Simple Groups}.
\newblock Springer, 2009.

\bibitem{deWolf19}
Ronald {\noopsort{Wolf}}{de Wolf}.
\newblock Quantum computing: Lecture notes.
\newblock ArXiv: 1907.09415, 2019.

\bibitem{Zhandry24}
Mark Zhandry.
\newblock Toward separating {QMA} from {QCMA} with a classical oracle.
\newblock In {\em Proceedings of the 16th Innovations in Theoretical Computer Science Conference (ITCS 2025)}, pages 95:1--95:19, 2025.

\end{thebibliography}
\end{document}